\documentclass[conference,a4paper]{IEEEtran}
\normalsize
\usepackage{graphicx}
\usepackage{subcaption}
\usepackage{mathtools, cuted}
\usepackage{array}
\usepackage[utf8]{inputenc}
\usepackage[T1]{fontenc}
\usepackage{amsmath}
\usepackage{mathtools}
\usepackage{amsfonts}
\usepackage{amssymb}
\usepackage{amsthm}
\usepackage{cite}
\usepackage{calc}
\usepackage{color}
\usepackage{epsfig}
\usepackage{setspace}
\usepackage{pstricks}
\usepackage{cancel}
\usepackage{multirow}
\usepackage{mathrsfs}
\usepackage[linesnumbered,ruled]{algorithm2e}
\usepackage{multirow}
\usepackage{diagbox}
\usepackage{booktabs,makecell}
\usepackage{mathrsfs}

\newtheorem{thm}{{\cal T}heorem}
\newtheorem{cor}{Corollary}

\newtheorem{lem}{Lemma}

\newtheorem{example}{Example}
\def\Ddots{\mathinner{\mkern1mu\raise\p@
		\vbox{\kern7\p@\hbox{.}}\mkern2mu
		\raise4\p@\hbox{.}\mkern2mu\raise7\p@\hbox{.}\mkern1mu}}
\makeatother
\newcommand{\ndiv}{\hspace{-4pt}\not|\hspace{2pt}}

\begin{document}
	\title{An Improved  Multi-access Coded Caching with Uncoded Placement}
	\author{
		\IEEEauthorblockN{Shanuja Sasi and B. Sundar Rajan}
		
		\IEEEauthorblockA{Dept. of Electrical Communication Engineering, \\Indian Institute of Science, Bengaluru \\
			E-mail: shanuja@iisc.ac.in, bsrajan@iisc.ac.in}
	}	
	\maketitle 
\begin{abstract}
		We consider a variant of the well known coded caching problem, referred as multi-access coded caching problem, where each user has access to $z$ neighboring  caches in a cyclic wrap-around way. We present a placement and delivery scheme for this problem, under the restriction of uncoded placement. Our work is a generalization of one of the cases considered in ``Multi-access coded caching : gains beyond cache-redundancy'' by B. Serbetci, E. Parrinello and P. Elia. To be precise, when our scheme is specialized to $z=\frac{K-1}{K \gamma }$, for any $K \gamma$, where $K$ is the number of users and $\gamma $ is the normalized cache size, we show that our result coincides with their result.
		We show that for the cases considered in this work, our scheme outperforms the scheme proposed in ``Rate-memory trade-off for multi-access coded caching with uncoded placement'' by  K. S. Reddy and N. Karamchandani, except for some special cases considered in that paper. We also show that for $z= K-1$, our scheme achieves the optimal transmission rate.
\end{abstract}
\section{introduction}
\label{sec1}
	The proliferation of users and their demand for high data rate content lead to a drastic increase in the high traffic volume during peak periods. In the seminal paper \cite{maddahali2015flcc}, Maddah-Ali and Nieson proposed a coded caching scheme to relieve the traffic burden during peak hours by utilizing the cache memories available at the user ends. The proposed scheme tackles the under utilization of resources during off-peak hours in the placement phase by filling the cache memories during the off-peak period so as to avail multicasting opportunities, when the demands are revealed by the users, during the delivery phase. The setup consists of a central server having access to a database comprising of a set of $N$ files of equal size and a set of $K$ users where each user has  a normalized cache size  $\gamma =\frac{M}{N},$ where $M$ is the memory size of each cache in the units of files.	Each user reveals their demand, which is assumed to be a single file among the files held by the server, during the peak hours. Then, the server transmits coded symbols to all the users over an error-free link such that each user can retrieve the demanded file from the local cache content and the transmitted symbols. The overall objective is to design the placement as well as the delivery phases with minimum transmission rate such that the user demands are satisfied. A lot of research has been done and quite a bit of variants of this problem have been studied during the past few years \cite{tuninetti2016uncoded,avestimar2017tradeoff,jin2018placeOpti,vilardebo2018coded,nieson2017nonuniform,gunduz2018nonuniform,nieson2015decent,caire2016d2d} where it is assumed that each user has their own dedicated cache. 
	
	However in several scenarios such as in cellular networks users can have access to multiple caches when their coverage areas overlap. Considering this possibility, recently a couple of studies \cite{nihkil2017multiaccess,nikhil2020ratememory,parinello2019multiacessgains, Caire2020noveltransformation} have been done  where each user has access to some $z$ neighboring caches in a cyclic wrap around fashion referred to as multi-access coded caching problem. Each cache is also connected to $z$ users.  In \cite{nikhil2020ratememory}, the authors have proposed a scheme for uncoded cache placement, by mapping of the coded caching problem to the index coding problem. A lower bound on the optimal transmission rate  for multi-access coded caching with $z \geq \frac{K}{2}$ over all uncoded placement policies was also provided in \cite{nikhil2020ratememory}. Additionally, the exact transmission rate-memory trade-off was established for a few special cases, i.e.,	when $z=K-1,z=K-2 $, $z=K-3$ with $K$ even and $z=K-\frac{K}{g}+1$ for some positive integer $g$.
	
	For the multi-access coded caching problem, in \cite{parinello2019multiacessgains} the authors  have proposed a novel coded caching scheme that achieves a coding gain that exceeds $K \gamma $.  Two special cases are considered in \cite{parinello2019multiacessgains}, one is when $K \gamma =2$ and the other is when $z=\frac{K-1}{K \gamma }$, for any $K \gamma $. For both the cases, the proposed scheme can serve, on the average, more than $K\gamma $ users at a time. For the second special case, i.e., when  $z=\frac{K-1}{K\gamma }$, it was proved that the achieved gain is optimal under uncoded cache placement.

{\it Notations:}  The finite field with $q$ elements is denoted  by $\mathcal{F}_q.$ The notation $[n]$ represents the set $\{1,2, \ldots , n\}$, $[a,b]$ represents the set $\{ a, a+1, \ldots, b \}$ and $[a {\text{ mod } K},b {\text{ mod } K}]$ represents the set $\{ a {\text{ mod } K}, (a+1) {\text{ mod } K}, \ldots, b {\text{ mod } K} \}$. The bit wise exclusive OR (XOR) operation is denoted by $\oplus.$ The notation  $\lfloor x \rfloor$ denotes the largest integer smaller than or equal to $x$ and $\lceil x \rceil$ denotes the smallest integer greater than or equal to $x$. 
The notation $a|b$ implies $a$ divides $b$ and 
$a \ndiv b$ implies $a$ does not divide $b$, for some integers $a$ and $b$.	
\subsection{Background and Preliminaries}
	In this section, we formally define the multi-access coded caching problem considered in this work. The system model \cite{nihkil2017multiaccess,nikhil2020ratememory}, as illustrated in Fig. \ref{model}, consists of a network comprising of a central server  storing $N$ files, $W^0,W^1,W^2,\ldots,W^{N-1}$, each of size $1$ unit, $K$ users, $U_0,U_1, \ldots,U_{K-1}$, and $K$ caches, $C_0,C_1, \ldots,C_{K-1}$, such that
	
	\begin{itemize}
		\item each user is connected to $z$ neighboring caches in a cyclic wrap-around fashion, and has access to the data stored in those caches, i.e., each user $U_{\alpha}, \alpha \in \{0,1,\ldots,K-1\}$, has access to all the caches in the set $\mathcal{C}_\alpha$, where  $\mathcal{C}_\alpha= \{C_{j \text{ mod } K}: j \in \{\alpha,\alpha+1,\ldots, \alpha+z-1\}\}.$
		\item each cache has a memory size of $M =N \gamma $ files, where $\gamma \in \{\frac{k}{K}, k=1,2,\cdots, K\}$ is the normalized cache size, 
		\item each user demands one among the $N$ files, and
		\item  all the $K$ users are connected via an error-free shared link to the server.
	\end{itemize}
%
	The system works in two phases- a placement phase, and	a delivery phase. In the placement phase the caches are filled with the content of the files from the servers' database prior to the users' requests. In the delivery phase, each user $U_{\alpha}$ demands a file from the database. The index of the file demanded by the user $U_{\alpha}$ is denoted by $d(\alpha)$. We denote ${\bf d} = (d(0) , d(1) , \ldots, d(K-1) )$ as the demand vector. When the demands are revealed by each user, the server transmits coded symbols depending upon the demand vector and cache content at each user. With the help of the server transmissions and the accessible cache content, each user $U_{\alpha}$ decodes the desired file $W^{d(\alpha)}$.  The multi-access coded caching problem is to develop placement and delivery schemes so as to minimize the transmission rate. 
\begin{figure*}[!t]
		\centering
		\includegraphics[width=24pc]{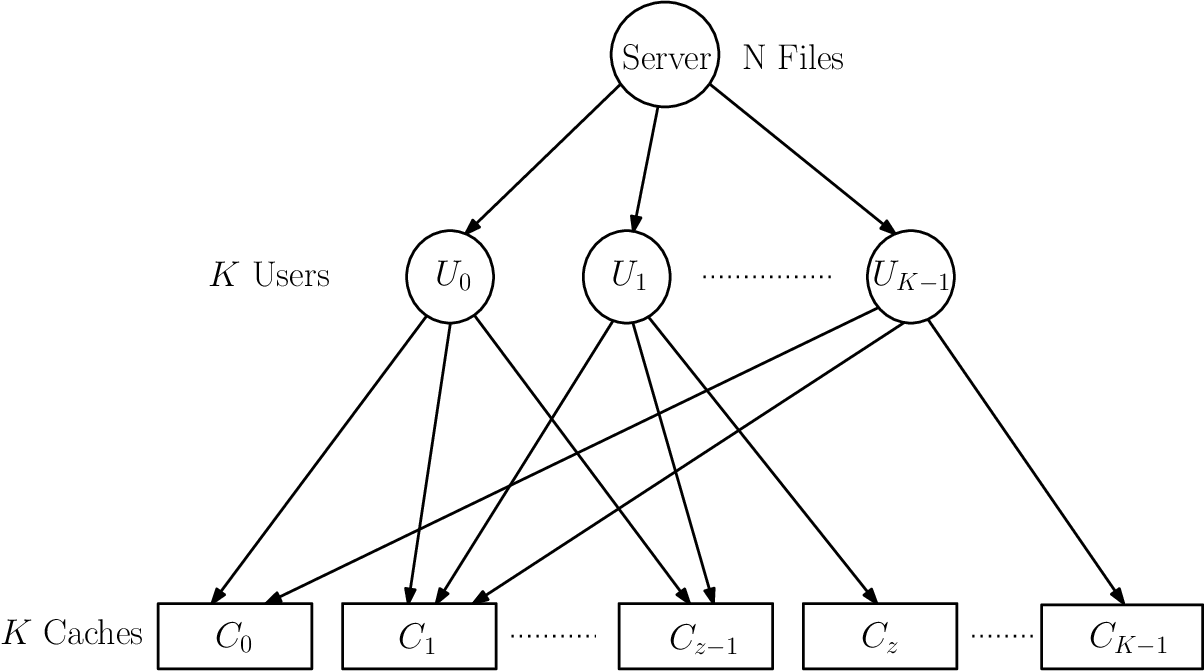}
		\caption{Multi-access Coded Caching Network \cite{nihkil2017multiaccess,nikhil2020ratememory} consisting of a central server, $K$ users, and $K$ caches where each user is connected to $z$ neighboring caches.}
		\label{model}
\end{figure*}
\subsection{Previous Results}
	The multi-access coded caching problem was introduced in \cite{nihkil2017multiaccess} where  the authors have proposed a coloring based achievable scheme. A new transmission rate $\mathcal{R}_{RK} $ was derived for multi-access coded caching problem with any $z>1$ in \cite{nikhil2020ratememory} which is better than that in \cite{nihkil2017multiaccess}.
{\small
\[
\mathcal{R}_{RK}(\gamma) = 
\left\{
\begin{array}{ll}
  K\left (1-z\gamma \right )^2, & \forall \gamma  \in \left \{0,\frac{1}{K},\frac{2}{K},\ldots, \left \lfloor \frac{K}{z} \right \rfloor \frac{1}{K} \right \} \\
   0, &  \textit{for     } \gamma = \left \lceil \frac{K}{z} \right \rceil \frac{1}{K}
\end{array}
\right.
\]}
The lower convex envelope of the above mentioned points is also achievable through memory sharing. A lower bound on the optimal transmission rate-memory trade-off for any multi-access coded caching problem under the restriction of uncoded placement with $z \geq \frac{K}{2}$ was derived in \cite{nikhil2020ratememory}. 
 The authors have also considered some special cases for $z \geq \frac{K}{2}$, namely when $z=K-1, z=K-2, z=K-3$ when $K $ is even, and
 $z=K-\frac{K}{g}+1$ for some positive integer $g$,
for which an achievable scheme was provided separately, which achieves the optimal transmission rate.

 In \cite{parinello2019multiacessgains}, the authors have considered two special cases, one is when $K\gamma =2$, and the other is when $z=\frac{K-1}{K\gamma }$ for any $K\gamma $. For $K\gamma =2$, a new scheme was proposed under certain conditions, which is better than the previous results. For the other special case considered in \cite{parinello2019multiacessgains}, i.e., when $z=\frac{K-1}{K\gamma }$ for any $K\gamma $, it was proved that the achieved transmission rate $\frac{1}{K}$ is optimal under uncoded cache placement.

%
\subsection{Our Contributions}
 The contributions of this paper are summarized as follows.
\begin{itemize}
	\item  We provide an achievable scheme for multi-access coded caching problem with each cache having a normalized capacity of $ \gamma$, where $\gamma \in \left \{\frac{k}{K}: \gcd(k,K)=1,  k\in \left [1,K \right ]\right \}$, under the restriction of uncoded placement.  
	\item 	Our work is a generalization of one of the cases considered in \cite{parinello2019multiacessgains}. To be precise, when our scheme is specialized to $z=\frac{K-1}{K\gamma }$, for any $K\gamma $, we show that our result coincides with that in Theorem $2$ in \cite{parinello2019multiacessgains}.
	\item 
	For the special cases, $z \geq \frac{K}{2}$, when $z=K-2, z=K-3$ when $K $ is even, and $z=K-\frac{K}{g}+1$ for some positive integer $g$, the scheme proposed in \cite{nikhil2020ratememory}, performs better than ours. The authors have provided separate optimal schemes for those special cases. For all other cases considered in our work, we show that our scheme outperforms the scheme proposed in \cite{nikhil2020ratememory}. We also show that for $z= K-1$, our scheme achieves the optimal transmission rate as in \cite{nikhil2020ratememory}.

\item In \cite{Caire2020noveltransformation} the authors have used a \textit{novel transformation} approach to multi-access caching schemes. 

It is proved that the rate achieved by this approach is less than the rate achieved using our scheme when $z \leq \left ( 1-\frac{2}{k}\right )\frac{K}{k}.$ We prove that our scheme performs better than the scheme proposed in  \cite{Caire2020noveltransformation} when $z \geq \left ( 1-\frac{1}{k}\right )\frac{K}{k}.$ 

\end{itemize}

The paper is organized as follows. Section \ref{sec2} provides the main result of this paper, i.e.,  Theorem \ref{thm1} presents the achievable transmission rate. We compare our result with the previous works in the same section. Our proposed scheme is described in Section \ref{sec3} which achieves the transmission rate presented in Theorem \ref{thm1}. Section \ref{sec4} concludes our paper. Proof of correctness of the delivery scheme is given in Appendix $A.$

\section{Main Result}
\label{sec2}	
We discuss our main result stated as Theorem \ref{thm1} in this section. 
\begin{thm} 
\label{thm1}
Consider a  multi-access coded caching scenario with $N$ files, and $K$ users, each having access to $z $ neighboring caches in a cyclic wrap-around way, with each cache having a normalized capacity of $ \gamma$, where $\gamma \in \left \{\frac{k}{K}: \gcd(k,K)=1,  k\in \left [1,K \right ]\right \}$. The following transmission rate $\mathcal{R}_{new}(\gamma)$ is achievable, for $z \in \left [2, \left \lceil \frac{K}{k} \right \rceil -1 \right ]$.
\begin{itemize}
\item  If  $(K- k z)=1,$  then $\mathcal{R}_{new}(\gamma) = \frac{1}{K}$.
\item If $(K- k z)$ is even, then 
$\mathcal{R}_{new}(\gamma) = 2 \sum_{r=\frac{K- k z}{2}+1}^{K- k z}\frac{1}{1+\left \lceil \frac{k z}{r} \right \rceil}.$
\item If $(K- kz)>1$ is odd, then 
{\small
$\mathcal{R}_{new}(\gamma) = \frac{1}{ \left (\left \lceil  \frac{2k z}{K-k z+1} \right \rceil+1 \right )}+\sum_{r=\frac{K- k z+3}{2}}^{K- k z}\frac{2}{1+\left \lceil \frac{k  z}{r} \right \rceil}. $}
\end{itemize}
\end{thm}
Note that if $z \geq \left \lceil \frac{K}{k} \right \rceil $, then all the users can access all the sub-files of each file and the transmission rate is zero. So, in Theorem \ref{thm1}, we consider only the cases where the transmission rate is non-zero.
The placement scheme and the delivery algorithm achieving the rate claimed in Theorem \ref{thm1} is given in Section \ref{sec3}.
\subsection{Comparison of our scheme with the scheme in \cite{nikhil2020ratememory}}
Theorem \ref{thm1} is illustrated in Fig. \ref{fig2} for $K =25$ and it is compared with the scheme proposed in \cite{nikhil2020ratememory}. In Fig. \ref{fig2}, transmission rate vs $\gamma$  plot is obtained for each $z \in \left [2, \left \lceil \frac{1}{\gamma} \right \rceil -1 \right ]$, as $\gamma $ varies from $\frac{1}{25} $ to $1$. It can be observed from the plot that our scheme performs better than that in \cite{nikhil2020ratememory} except for $z=K-2$. For $z=K-2=23$ $ (\frac{z}{K}=\frac{23}{25})$, the scheme in \cite{nikhil2020ratememory} outperforms ours, since the case of $z=K-2$ was considered separately in \cite{nikhil2020ratememory} and an optimal scheme for that particular case was provided in \cite{nikhil2020ratememory}. For each $z$, the gap between the two curves is large for smaller $\gamma$. As $\gamma $ increases, the gap reduces and both the curves coincide eventually.

Now, we examine the case when $k=1$. A plot for $K =11$ and $\gamma =\frac{1}{11}$ is shown in Fig. 3(a). In this plot, it can be seen that our scheme is better than the scheme in \cite{nikhil2020ratememory} except for one point in Fig. 3(a). The gap between the transmission rate of our scheme and that in \cite{nikhil2020ratememory} is more for smaller $z$. As $z$ increases, the gap reduces and gradually both the curves coincide. The only point where the scheme in \cite{nikhil2020ratememory} is better that ours is when $z=K-2=9$ in Fig. 3(a). Like it was discussed before, this is since the case of $z=K-2$ was considered separately in \cite{nikhil2020ratememory} and an optimal scheme for that particular case was provided in \cite{nikhil2020ratememory}.

In general, for all the points mentioned in Theorem \ref{thm1}, our scheme outperforms the scheme proposed in \cite{nikhil2020ratememory} except for some special cases discussed in \cite{nikhil2020ratememory} which achieves the optimal transmission rate. This is shown in Theorem \ref{thm2}. And also, when $z=K-1$, our scheme achieves the optimal transmission rate which coincides with the result in \cite{nikhil2020ratememory}. This particular case is discussed in Corollary \ref{cor2}. 
\begin{figure*}[!t]
	\centering
	\includegraphics[scale=0.3]{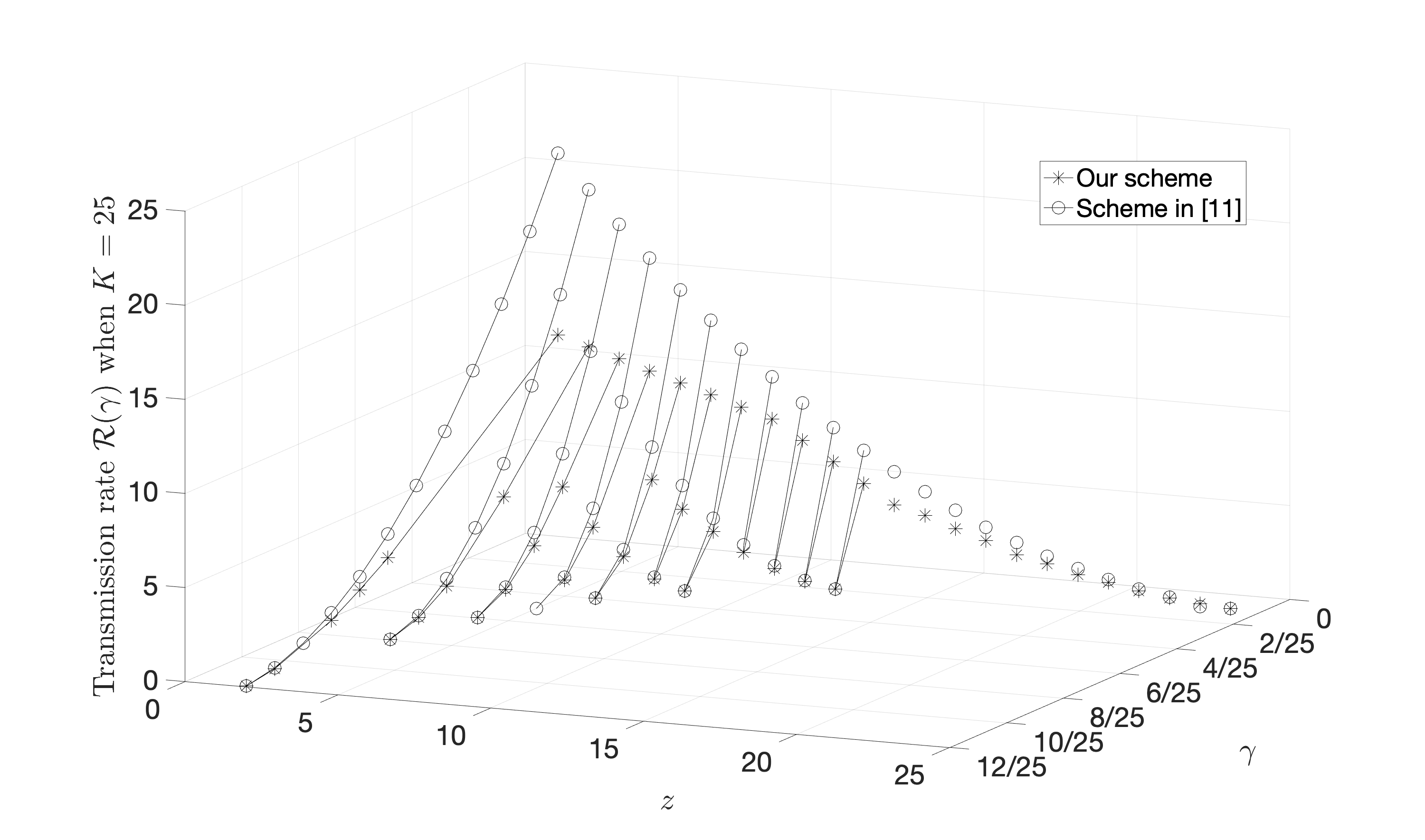}
	\caption{Transmission rate vs $\gamma$ vs $z$ plot when $K =25$, as in Theorem \ref{thm1}.}
	\label{fig2}
\end{figure*}
\begin{figure*}
	\begin{subfigure}{0.49\textwidth}
		\centering
		\includegraphics[scale=0.17]{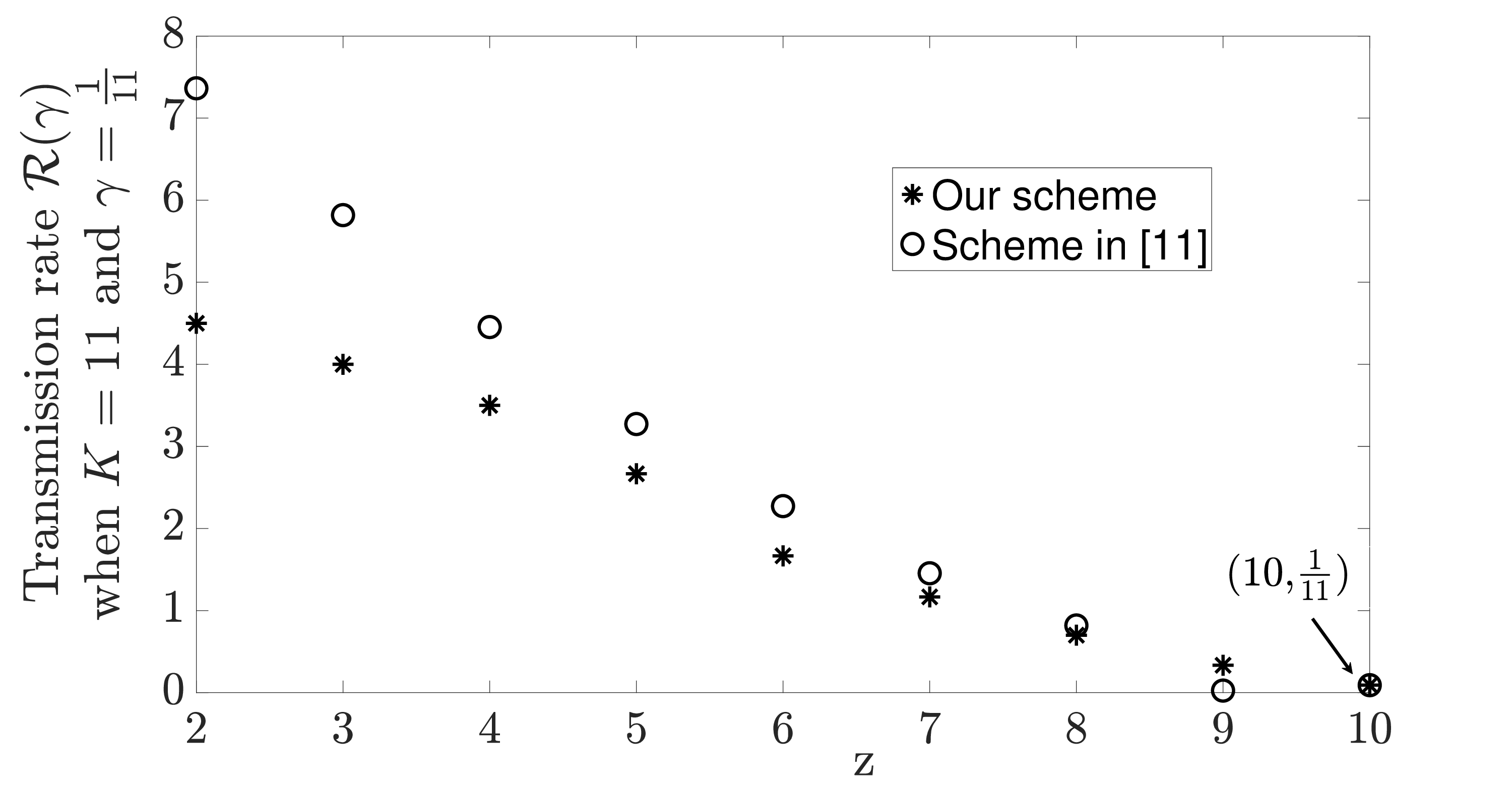}
		\caption{Transmission rate vs $z$ when $k =1$, for $K=11$ as in Theorem \ref{thm1}.}
		\label{fig3}
	\end{subfigure}
\begin{subfigure}{0.49\textwidth}
	\centering
	\includegraphics[scale=0.16]{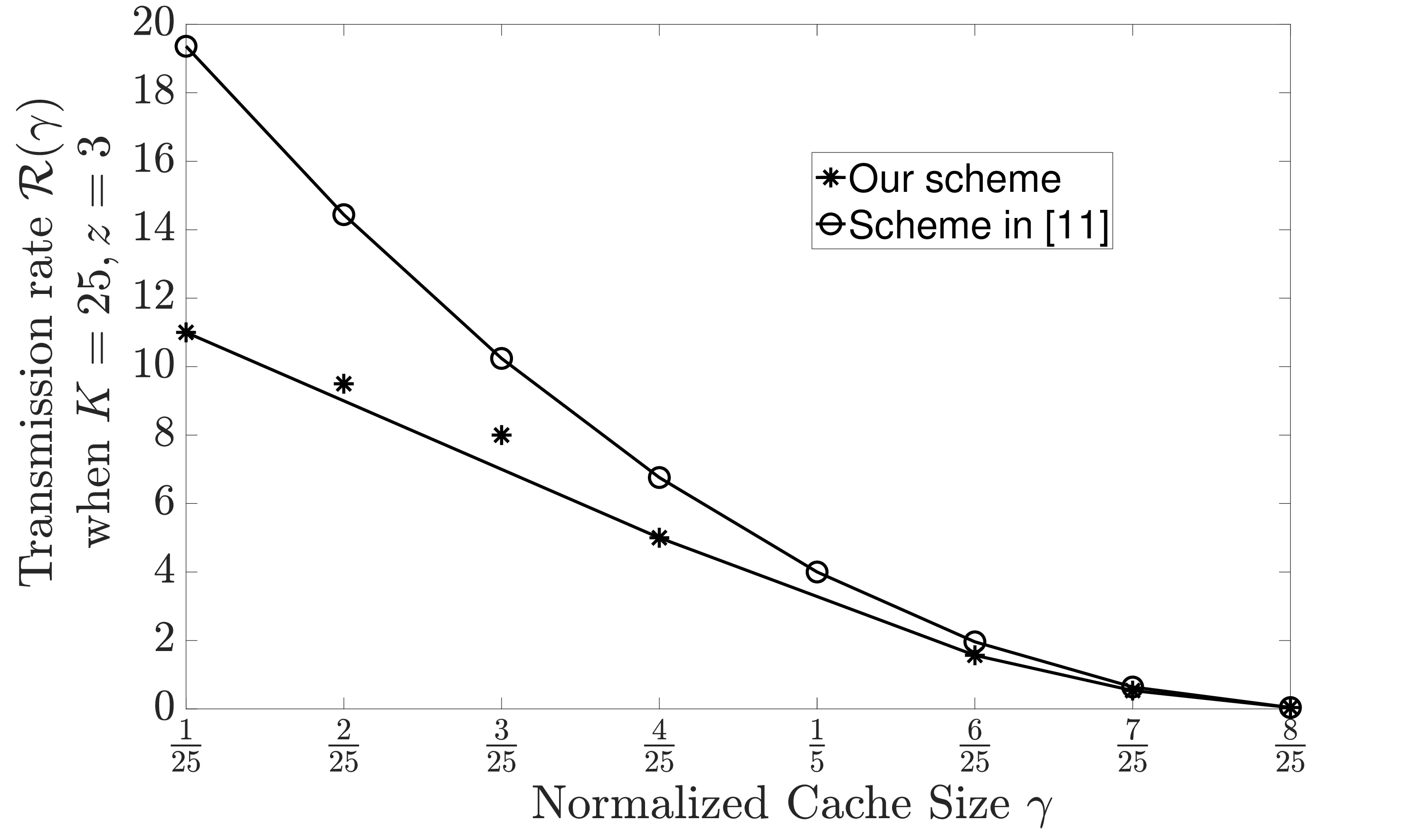}
	\caption{Transmission rate vs $\gamma$ plot when $K =25$ and $z=3$ for all the points mentioned in Theorem \ref{thm1}.}
	\label{fig5}
\end{subfigure}
	\caption{ Transmission rate vs $z$ plot and transmission rate vs $\gamma$ plot  as in Theorem \ref{thm1}.}
\end{figure*}
\begin{thm}
\label{thm2}
For any $\gamma \in \left \{\frac{k}{K}: \gcd(k,K)=1,  k\in \left [1,K \right ]\right \}$, we have  $\mathcal{R}_{new}(\gamma) \leq \mathcal{R}_{RK}(\gamma).$
\end{thm}
\begin{proof}
	For any $\gamma$, we have $\mathcal{R}_{RK}\left (\gamma \right ) =K\left (1-\frac{zk }{K}\right )^{2}=  	\frac{(K-k z)^2}{K}$.  

{\it Case (i):} When $K-kz=1$, we have  $\mathcal{R}_{RK}(\gamma ) =\frac{(K-k z)^2}{K} =\frac{1}{K} =\mathcal{R}_{new}(\gamma).$

{\it Case (ii):}  When $K-k z$ is even, for any $i  \in  \{1,2,\ldots,\frac{K-k z}{2}\}$, we have 
{\small
\begin{align*}
	\frac{k z}{K-k z} \leq \frac{k z}{\frac{K-k z}{2}+i} &\leq \left \lceil \frac{k z}{\frac{K-k z}{2}+i} \right \rceil \\
	\Rightarrow 1+\frac{k z}{K-k z} &\leq 1+ \left \lceil \frac{k z}{\frac{K-k z}{2}+i} \right \rceil. 
	\end{align*}}
For any $r \in \left [\frac{K-k z}{2}+1,K-k z \right ]$, 
{\small
\begin{align*}
 \frac{K}{K-k z} &\leq 1+ \left \lceil \frac{k z}{r} \right \rceil   \\
	\Rightarrow \frac{K-k z}{K} &\geq \frac{1}{1+ \left \lceil \frac{k z}{r} \right \rceil} \\
	\Rightarrow \sum_{r=\frac{K- k z}{2}+1}^{K-k z} \frac{K-k z}{K}  &\geq \sum_{r=\frac{K- k z}{2}+1}^{K-k z}   \frac{1}{1+ \left \lceil \frac{k z}{r} \right \rceil}\\
	\Rightarrow \left ( \frac{K- k z}{2} \right ) \left ( \frac{K-k z}{K} \right ) & \geq \sum_{r=\frac{K- k z}{2}+1}^{K-k z}   \frac{1}{1+ \left \lceil \frac{k z}{r} \right \rceil} \\
	\Rightarrow   \frac{(K-k z)^2}{K}  &\geq 2\sum_{r=\frac{K- k z}{2}+1}^{K-k z}   \frac{1}{1+ \left \lceil \frac{k z}{r} \right \rceil} 
	\end{align*}	}
	Hence, if $K-k z$ is even, then $ \mathcal{R}_{new}\left (\gamma \right ) \leq \mathcal{R}_{RK}\left (\gamma \right )$.

{\it Case (iii):} When $K-k z >1$ is odd, for any $i  \in  \{1,2,\ldots,\frac{K-k z+1}{2}\}$, we have 
{\small
\begin{align*}
	\frac{k z}{K-k z} \leq \frac{k z}{\frac{K-k z-1}{2}+i} &\leq \left \lceil \frac{k z}{\frac{K-k z-1}{2}+i} \right \rceil \\
	\Rightarrow 1+\frac{k z}{K-k z} &\leq 1+ \left \lceil \frac{k z}{\frac{K-k z-1}{2}+i} \right \rceil \\ 
	\Rightarrow \frac{K}{K-k z} &\leq 1+ \left \lceil \frac{k z}{\frac{K-k z-1}{2}+i} \right \rceil.   
\end{align*}}
For any $r \in \left [\frac{K-k z+1}{2},K-k z   \right ]$, we have
{\small
\begin{align} \label{eq: 1}
 \frac{K}{K-k z} \leq 1+ \left \lceil \frac{k z}{r} \right \rceil   
	\Rightarrow \frac{K-k z}{K} \geq \frac{1}{1+ \left \lceil \frac{k z}{r} \right \rceil } .
\end{align}}
Taking sum over all $r \in \left \{\frac{K-k z+1}{2}+1,\frac{K-k z+1}{2}+2,\ldots,K-k z \right \}$, we get 
	{\small
	\begin{align}
	\sum_{r=\frac{K- k z+1}{2}+1}^{K-k z} \frac{K-k z}{K}   &\geq \sum_{r=\frac{K- k z+1}{2}+1}^{K-k z}   \frac{1}{1+ \left \lceil \frac{k z}{r} \right \rceil} \\
	\Rightarrow \left ( \frac{K- k z-1}{2} \right ) \left ( \frac{K-k z}{K}  \right )& \geq \sum_{r=\frac{K- k z+1}{2}+1}^{K-k z}   \frac{1}{1+ \left \lceil \frac{k z}{r} \right \rceil}\\
	\Rightarrow   \frac{(K-k z-1)  (K-k z)}{K} &  \geq 2\sum_{r=\frac{K- k z+1}{2}+1}^{K-k z}   \frac{1}{1+ \left \lceil \frac{k z}{r} \right \rceil}. \label{eq: 2}
\end{align}}

	Now considering the inequality (\ref{eq: 1}) when $r = \frac{K-k z+1}{2},$ we get 
	{\small
\begin{align} \label{eq: 3}
	\frac{K-k z}{K} &  \geq   \frac{1}{1+ \left \lceil \frac{k }{\frac{K-k  z+1}{2}} \right \rceil}.
\end{align}}
	Summing up the inequalities $(\ref{eq: 2})$ and $(\ref{eq: 3})$, we get,
	{\small
\begin{align*}
	\frac{K-k z}{K}  + \frac{(K-k z-1) (K-k z)}{K}  \geq\\  \frac{1}{1+ \left \lceil \frac{k z}{\frac{K-k  z+1}{2}} \right \rceil}+2\sum_{r=\frac{K- k z+1}{2}+1}^{K-k z}   \frac{1}{1+ \left \lceil \frac{k z}{r} \right \rceil}. 
\end{align*}}
	Hence,
\begin{align*}
	\frac{(K-k z)^2}{K}   &\geq  \frac{1}{1+ \left \lceil \frac{k z}{\frac{K-k  z+1}{2}} \right \rceil}+2\sum_{r=\frac{K- z+1}{2}+1}^{K-z}   \frac{1}{1+ \left \lceil \frac{z}{r} \right \rceil} \\
	\Rightarrow \mathcal{R}_{RK}\left (\gamma \right ) & \geq  \mathcal{R}_{new}\left (\gamma \right ).
\end{align*}	
\end{proof}
\begin{cor}
	\label{cor2}
	For $z=K-1$, we have $\mathcal{R}^{*}(\gamma) = \mathcal{R}_{new}(\gamma)$, where $R^{*}(\gamma)$ is the optimal rate of dedicated-cache coded caching setting
	where each cache has an augmented size equal to $z\gamma.$
\end{cor}
\begin{proof}
If $kz > K-1,$ then the transmission rate is zero, since each user can access all the sub-files of all the files. If $k \geq 2$, then  $kz=2(K-1) > K-1$ and hence  the transmission rate is zero. When $k =1$, we have  $K-kz=1$ and hence $\mathcal{R}_{new}\left( \gamma \right ) = \frac{1}{K}=\mathcal{R}^{*}\left( \gamma \right).$
\end{proof}
\subsection{Comparison of our scheme with the scheme in \cite{parinello2019multiacessgains}}
In Fig. 3(a), one specific point $(K-1,\frac{1}{K})$ is marked  which corresponds to one of the cases discussed in \cite{parinello2019multiacessgains}, where $z=\frac{K-1}{k}$. Considering this case, in Fig. 3(a), since $k =1$, we have $z=K-1$, for which a scheme was proposed in \cite{parinello2019multiacessgains} which achieves the optimal transmission rate $\frac{1}{K}$. The transmission rate obtained from our scheme coincides with that.  In general, when our scheme is specialized to $z=\frac{K-1}{k}$,  for any $k $, our result coincides with that of Theorem $2$ in \cite{parinello2019multiacessgains}. This is because $\mathcal{R}_{new}(\gamma )  = \frac{1}{K}$, as $K- k z = K -(K-1) =1$. Hence, $\mathcal{R}^{*}(\gamma) =\mathcal{R}_{new}(\gamma)=\frac{1}{K}$.
%

\subsection{Comparison of our scheme with the scheme in \cite{Caire2020noveltransformation}}
After the initial submission of this manuscript the paper \cite{Caire2020noveltransformation} appeared in which the authors have used a \textit{novel transformation} approach to multi-access caching schemes. The coding gain achieved using  the novel transformation approach is $K\gamma +1$ with a sub-packetization level of $K {K-K\gamma(L-1) \choose K\gamma}$.
It is proved in \cite{Caire2020noveltransformation} that the rate achieved in \cite{Caire2020noveltransformation} is less than the rate achieved using our scheme when $z \leq \left ( 1-\frac{2}{k}\right )\frac{K}{k}.$ We prove that our scheme performs better than the scheme proposed in  \cite{Caire2020noveltransformation} when $z \geq \left ( 1-\frac{1}{k}\right )\frac{K}{k}$ as follows:

 If $(K- k z)$ is even, then the rate achieved using our scheme is
\begin{align*}
\mathcal{R}_{new}(\gamma) &= 2 \sum_{r=\frac{K- k z}{2}+1}^{K- k z}\frac{1}{1+\left \lceil \frac{k z}{r} \right \rceil} \\ &
\leq 2\left (\frac{K-kz}{2} \right )\frac{1}{1+\left \lceil \frac{k z}{K-kz} \right \rceil}\\
&\leq \frac{K-kz}{1+\left \lceil \frac{k z}{K-kz} \right \rceil}.
\end{align*}
Similarly, if $(K- kz)$ is odd, then the rate achieved using our scheme is

{\small
\begin{align*}
\mathcal{R}_{new}(\gamma) 
&= \frac{1}{ \left (\left \lceil  \frac{2k z}{K-k z+1} \right \rceil+1 \right )}+\sum_{r=\frac{K- k z+3}{2}}^{K- k z}\frac{2}{1+\left \lceil \frac{k  z}{r} \right \rceil}.\\
& \leq  \frac{1}{ \left (\left \lceil  \frac{2k z}{K-k z+1} \right \rceil+1 \right )}+  2\left (\frac{K-kz-1}{2} \right )\frac{1}{1+\left \lceil \frac{k z}{K-kz} \right \rceil}\\
& \leq  \frac{1}{ \left (\left \lceil  \frac{k z}{K-k z} \right \rceil+1 \right )}+  \left (K-kz-1 \right )\frac{1}{1+\left \lceil \frac{k z}{K-kz} \right \rceil} \\& \leq \frac{K-kz}{1+\left \lceil \frac{k z}{K-kz} \right \rceil}.
\end{align*}}
Therefore, we have, $
\frac{R_{new}(\gamma)}{R_{NT}(\gamma)} \leq \left (\frac{K-kz}{1+\left \lceil \frac{k z}{K-kz} \right \rceil} \right ) \left (\frac{k+1}{K-kz}\right ) \leq \frac{k+1}{1+\left \lceil \frac{k z}{K-kz} \right \rceil}.$ Hence, $R_{new}(\gamma) \leq R_{NT}(\gamma)$, if
\begin{align*}
\frac{k+1}{1+\left \lceil \frac{k z}{K-kz} \right \rceil} &\leq 1\\ \Rightarrow k+1  &\leq 1+\left \lceil \frac{k z}{K-kz} \right \rceil\\ \Rightarrow k&\leq \left \lceil \frac{k z}{K-kz} \right \rceil\\ \Rightarrow k &\leq \frac{k z}{K-kz} +1\\
\Rightarrow k&\leq \frac{K}{K-kz}\\ \Rightarrow K-kz &\leq \frac{K}{k} \\\Rightarrow K- \frac{K}{k} &\leq kz \\
\Rightarrow z &\geq \frac{K}{k}\left (1- \frac{1}{k}\right ).\\
\end{align*}

\section{placement and delivery scheme}
\label{sec3}
In this section we present our placement and delivery scheme to prove  Theorem \ref{thm1}. 
\subsection{Placement Scheme}
	In the placement phase, we split each file $W^n, n=\{0,1,\ldots,N-1\}$, into $K $ disjoint sub-files $W_{\alpha}^n, \alpha \in \{0,1,\ldots,K-1\}$. Let $M_{\alpha}$ represent the content stored in the cache $C_{\alpha}, \alpha \in \{0,1,\ldots, K-1 \}$. Each cache $C_{\alpha}$ is filled as follows:
	{\small$M_{\alpha}= \{W_{(k \alpha +j) \text{ mod } K}^n:  j \in \{0,1,\ldots ,k -1\},   n \in \{0,1,\ldots,N-1\}\}.$} Each cache stores $k $ sub-files from all the files, where each sub-file is of size $\frac{1}{K}$. Hence, $M=\frac{k N}{K} =N \gamma $, thus meeting our memory constraint.
	
	In this phase, the sub-files of each file are placed in such a way that we first create a list of size $ kK$ by repeating the sequence $\{0, 1, \ldots, K-1 \}$, $k$ times, i.e., $\{0,1, \ldots,K-1,0,1,\ldots,K-1,0,1,2,\ldots, K-1,\ldots\}.$ Then we fill the caches by taking $k$ items sequentially from the list, where each item on the list corresponds to the index of the sub-file. So, the first cache is filled with the first $k$ items, the second cache with the next $k$ items and so on. We put a restriction on the value of $k$ (that $\gcd(k,K)=1$) to make sure that if we take any two caches, they do not store the same $k$ sub-files. This is since we use this property in our delivery algorithm.
	
		Each user can access $z$ neighboring caches and each cache stores $k$ consecutive sub-files of each file. 
	We take only the case when $z < \left \lceil\frac{K}{k} \right \rceil$. This is because, if $z \geq \left \lceil \frac{K}{k} \right \rceil$, then the user has access to all the sub-files of each of the files. Hence, for the case under consideration, each user has access to $kz$ consecutive sub-files of each file since the content in any consecutive $z$ caches are disjoint from one another if $z < \left \lceil\frac{K}{k} \right \rceil$.
	That is, for each user $U_{\alpha} , \alpha\in [0,K-1]$, the  accessible cache content is $\{W_{(k \alpha +i) \text{ mod } K}^{n} : i \in [0, kz-1],   n \in [0,N-1]\}$.
	 \begin{algorithm*}
	
	\caption{Delivery scheme for multi-access coded caching if $\gamma \in \left \{\frac{k}{K}: \gcd(k,K)=1,  k\in \left [1,K \right ]\right \}$.}
	\label{algo2}

	\For{iteration $i=1,2,\ldots, \left \lceil \frac{K-k z}{2} \right \rceil $}{
		Let $r =\left \lfloor \frac{K-k z}{2} \right \rfloor +i$ and $p=\left \lceil \frac{k z}{r} \right \rceil+1$.\\
		Define two one-to-one and onto functions, on $ [0,K-1]$: $\pi_{i,1 }(\alpha) =\left (k\alpha-\left \lfloor \frac{K-k z}{2} \right \rfloor -i \right ) \text{ mod } K=(k\alpha-r) \text{ mod } K$ and $\pi_{i,2 }(\alpha) =\left (k\alpha+kz+\left \lfloor \frac{K-k z}{2} \right \rfloor +i-1 \right ) \text{ mod } K=(k\alpha+kz+r-1) \text{ mod } K$.
		
		\eIf{$(K-kz)$ is odd and $i=1$}
		{	
			Split each sub-file in $\left \{ W_{\alpha}^{d \left (\pi^{-1}_{{i,1}}(\alpha) \right )}: \alpha \in [0,K-1] \right \}$  into $p$ blocks: $\left \{ W_{\alpha,l}^{d \left (\pi^{-1}_{{i,1}}(\alpha) \right )}: \alpha \in [0,K-1], l \in[0,p-1] \right \}$.
			
			\For{each $j \in \{0,1, \ldots , K-1\}$} {Transmit one coded symbol 
				$T_j^i = \bigoplus_{m=0}^{p-1} W_{(mr+j) \text{ mod } K,m}^{d \left (\pi^{-1}_{{i,1}}((mr+j) \text{ mod } K) \right ) }$
		}}
		{
			
			\eIf{$p$ is even}
			{Split each sub-file in $\left \{ W_{\alpha}^{d\left(\pi^{-1}_{i ,1}(\alpha)\right)},  W_{\alpha}^{d\left(\pi^{-1}_{i,2 }(\alpha)\right )}: \alpha \in [0,K-1] \right \}$ into $\frac{p}{2} $ blocks: {\small $\left\{W_{\alpha,l}^{d\left(\pi^{-1}_{{i,1}}(\alpha)\right)}, W_{\alpha,l}^{d\left(\pi^{-1}_{{i,2}}(\alpha)\right)}: \alpha \in [0,K-1], l \in \left [0, \frac{p}{2}-1  \right ] \right\}$.}
				
				\For{each $j \in \{0,1, \ldots , K-1\}$} {
					Transmit one coded symbol $T_j^{i}$:
					\begin{align*}
					T_j^{i} = \bigoplus_{m =0}^{\frac{p}{2}-1} \left (W_{(mr+j) \text{ mod } K,m}^{d\left(\pi^{-1}_{{i,1}}((mr+j) \text{ mod } K)\right)  }
					\oplus W_{\left (\left (m+\frac{p}{2} \right )r+j \right ) \text{ mod } K,m}^{d\left(\pi^{-1}_{{i,2}}\left(\left (\left (m+\frac{p}{2} \right )r+j \right ) \text{ mod } K\right )\right)  } \right ) 
					\end{align*}
			}} 
			{Split each sub-file in $\left \{ W_{\alpha}^{d\left(\pi^{-1}_{i,1}(\alpha)\right)},  W_{\alpha}^{d\left(\pi^{-1}_{i ,2}(\alpha)\right)}: \alpha \in [0,K-1] \right \}$  into $p $ blocks: {\small $\left\{W_{\alpha,l}^{d\left(\pi^{-1}_{{i,1}}(\alpha)\right)}, W_{\alpha,l}^{d\left(\pi^{-1}_{{i,2}}(\alpha)\right )}: \alpha \in [0,K-1], l \in \left [0, p-1  \right ] \right\}$.}
				
				\For{each $j \in \{0,1, \ldots , K-1\}$} 
				{Transmit two coded symbols:
					\begin{align*}
					T_{j,1}^{i} = \left (\bigoplus_{m=0}^{\frac{p-3}{2}} W_{(mr+j) \text{ mod } K,m}^{d\left (\pi^{-1}_{{i,1}}((mr+j) \text{ mod } K)\right) } \right ) \bigoplus \left (
					\bigoplus_{m= \frac{p-1}{2}}^{p-1} W_{(mr+j) \text{ mod } K,m-\frac{p-1}{2}}^{d \left (\pi^{-1}_{{i,2}}((mr+j) \text{ mod } K)\right)  } \right )
					\end{align*}
					\begin{align*}
					T_{j,2}^{i} = \left (\bigoplus_{m =0}^{\frac{p-1}{2}} W_{(mr+j) \text{ mod } K,\frac{p-1}{2}+m}^{d\left(\pi^{-1}_{{i,1}}((mr+j) \text{ mod } K)\right) } \right ) \bigoplus \left (
					\bigoplus_{m= \frac{p+1}{2}}^{p-1} W_{(mr+j) \text{ mod } K,m}^{d\left(\pi^{-1}_{{i,2}}((mr+j) \text{ mod } K)\right)  } \right )
					\end{align*}
			}} 
	}} 
	
\end{algorithm*}

\subsection{Delivery Scheme}
	Each user's demand, of one file among the $N$ files from the central server, is revealed after the placement phase. Once the demand vector ${\bf d}$ is known, Algorithm \ref{algo2} provides the transmissions done by the server when  $\gamma \in \left \{\frac{k}{K}: \gcd(k,K)=1,  k\in \left [1,K \right ]\right \}$. 
	
	The parts of the file $W^{d(\alpha)}$ available with the user $U_{\alpha}$ are $kz$ consecutive sub-files, i.e.,  $\{W_{(k \alpha+i) \text{ mod } K}^{d(\alpha)} : i \in [0,kz-1]\}$. Hence,  the user $U_{\alpha}$ should be able to decode all the remaining $K-kz$ sub-files, i.e.,  $\{W_{(k \alpha+kz+i) \text{ mod } K}^{d(\alpha)} : i \in [0,K-k z-1]\}$ from the transmissions done using Algorithm \ref{algo2}.
	
	The coded symbols are obtained in $ \left \lceil \frac{K-k z}{2} \right \rceil $ iterations using Algorithm \ref{algo2}. With each iteration, $r$ takes values from $\left \lfloor \frac{K-k z}{2} \right \rfloor +1$ to $K-kz$ (as in line 2 of Algorithm \ref{algo2}). So, the value of $p$ ranges from $\left \lceil \frac{k z}{\left \lfloor \frac{K-k z}{2} \right \rfloor +1}  \right \rceil +1$ to $ \left \lceil \frac{k z}{K-kz} \right \rceil +1$. That is, during iteration $i$, the value of $p$ is chosen as $\left \lceil \frac{k z}{\left \lfloor \frac{K-k z}{2} \right \rfloor +i}  \right \rceil +1$.	Hence, $p$ is a non-increasing function of $i$.
 	For a given $k$ and $z$, if $kz < \frac{K}{2}$, then $\left \lceil \frac{k z}{\left \lfloor \frac{K-k z}{2} \right \rfloor +i}  \right \rceil =1, $ for each $i \in \left [1, \left \lceil \frac{K-k z}{2} \right \rceil \right ]$. Hence, the value of $p$ is always $2$ for each of the iterations. If $kz \geq \frac{K}{2}$, the maximum value $p$ can take is when $kz$ is maximum, i.e., $kz=K-1$. When $kz=K-1$, the value that  $p$ takes during iteration $i$ is $\left \lceil \frac{K-1}{i}  \right \rceil +1$. So, when $kz=K-1$, the maximum value that $p$ takes	is $K $ (when $i=1$).  Hence, the value of $p$ ranges from $2$ to $K$.

	In  line 3 of Algorithm \ref{algo2}, we define two functions $\pi_{i,1 }(\alpha)$ and $\pi_{i,2 }(\alpha)$. We have $\{k\alpha \text{ mod } K: \alpha\in [0,K-1]\} =\{0,1, \ldots , K-1\}$, since we assume that $\gcd(k,K)=1$. Hence $\{ (k\alpha -r ) \text{ mod } K : \alpha \in [0,K-1] \}=\{0,1, \ldots , K-1\}$.  Similarly, for the same reason,  we have $\{(k\alpha+kz+r-1 )  \text{ mod } K: \alpha\in [0,K-1]\} =\{0,1, \ldots , K-1\}$.
	  Therefore, we have two one-to-one and onto functions
$\pi_{i,1 }(\alpha) =(k\alpha-r) \text{ mod } K$ and $\pi_{i,2 }(\alpha) =(k\alpha+kz+r-1) \text{ mod } K$, on $ [0,K-1]$.

Now we look into the transmission rate involved in this scheme. If $K-k z $ is even, from Algorithm \ref{algo2}, for each $r \in \left \{\frac{K-k z}{2}+1,\frac{K-k z}{2}+2,\ldots, K-kz \right \}, $ the amount of transmission done by the server is $\frac{2K}{p} \left ( \frac{1}{K} \right)=\frac{2}{1+\left \lceil \frac{k z}{r} \right \rceil}$ files, accounting for a  total transmission rate of $ \mathcal{R}_{new}(\gamma) = 2 \sum_{r=\frac{K- k z}{2}+1}^{K- k z}\frac{1}{1+\left \lceil \frac{k z}{r} \right \rceil}$. Similarly considering the case when $K-k z >1$ and  is odd, from Algorithm \ref{algo2}, for each $r \in \{\frac{K-k z+1}{2}+1,\ldots, K-kz\}, $ the amount of transmission done by the server is $\frac{2K}{p} \left ( \frac{1}{K} \right)=\frac{2}{1+\left \lceil \frac{k z}{r} \right \rceil}$ files. If $r=\frac{K-k z+1}{2}$, the amount of transmission done by the server is $\frac{K}{p} \left ( \frac{1}{K} \right)=\frac{1}{1+\left \lceil \frac{k z}{r} \right \rceil}$ files. Hence the overall transmission rate is $	 \mathcal{R}_{new}(\gamma) = 2 \left ( \frac{1}{2 \left (\left \lceil  \frac{2k z}{K-k z+1} \right \rceil+1 \right )}+\sum_{r=\frac{K- k z+3}{2}}^{K- k z}\frac{1}{1+\left \lceil \frac{k z}{r} \right \rceil } \right )$. If $K-k z =1$,  the amount of transmission done by the server is $\frac{K}{p} \left ( \frac{1}{K} \right)=\frac{1}{1+\left \lceil \frac{k z}{t} \right \rceil} = \frac{1}{1+\left \lceil \frac{2(k z)}{K-k z +1} \right \rceil } =\frac{1}{1+(K-1)} =\frac{1}{K}$ files. Hence the transmission rate is $\mathcal{R}_{new}(\gamma) = \frac{1}{K}$. 
The detailed proof of the delivery scheme, i.e., the proof of decoding  using Algorithm \ref{algo2} is given in Appendix A. 

Now, we illustrate the idea of our proposed placement and delivery schemes using the following examples.
\begin{table*}[!ht]
	\centering
	\begin{tabular}{ |c|c|c|c|c|c|} 
		\hline
		&{\small \textit{ $W_0^n$} }&{\small \textit{ $W_1^n$}}&{\small \textit{ $W_2^n$} }&{\small \textit{$W_3^n$}}&{\small  \textit{ $W_4^n$} } \\
		\hline \hline
		{\small $T_{0}^1$} &\textcolor{red}{{\small   $(0,d(2))$}}&  &\textcolor{blue}{{\small $(1,d(4))$ }}&&    \\ \hline
		{\small $T_{1}^1$} & & \textcolor{red}{{\small   $(0,d(3))$}}& &\textcolor{blue}{{\small $(1,d(0))$}}& \\ \hline
		{\small $T_{2}^1$ }& &   &\textcolor{red}{ {\small $(0,d(4))$}}&&\textcolor{blue}{{\small $(1,d(1))$ }}   \\ \hline
		{\small $T_{3}^1$ }&\textcolor{blue}{{\small $(1,d(2))$  }}&   & &\textcolor{red}{{\small $(0,d(0))$}}&  \\ \hline
		{\small $T_{4}^1$ }& & \textcolor{blue}{{\small $(1,d(3))$ }} & &&\textcolor{red}{{\small $(0,d(1))$ } } \\ \hline
	\end{tabular}
	\caption{ {\small Table that illustrates the sub-files included in the coded symbols obtained in  the first iteration in Example \ref{exmp1}.} 
	}
	\label{tab eg1 iter 1}
\end{table*}
\begin{table*}[!ht]
	\centering
	\begin{tabular}{ |c|c|c|c|c|c|} 
		\hline
	&{\small  \textit{ $W_0^n$} }&{\small  \textit{ $W_1^n$}}&{\small \textit{ $W_2^n$} }&{\small \textit{$W_3^n$}}& {\small \textit{ $W_4^n$}} \\
		\hline \hline
	{\small 	$T_{0}^2$ }&\textcolor{red}{{\small   $(0,d(3))$}}&  & &\textcolor{blue}{{\small  $(0,d(4))$}}&  \\ \hline
	{\small $T_{1}^2$} & & \textcolor{red}{{\small   $(0,d(4))$}}& &&\textcolor{blue}{{\small  $(0,d(0))$}} \\ \hline
    {\small 	$T_{2}^2$} & \textcolor{blue}{{\small $(0,d(1))$}}&   & \textcolor{red}{{\small $(0,d(0))$}}&&  \\ \hline
    {\small $T_{3}^2$ }&  &\textcolor{blue}{{\small  $(0,d(2))$ }} & &\textcolor{red}{{\small $(0,d(1))$}}& \\ \hline
    {\small $T_{4}^2$ }& &  &\textcolor{blue}{{\small  $(0,d(3))$} }&&\textcolor{red}{{\small $(0,d(2))$}}  \\ \hline
	\end{tabular}
	\caption{ {\small Table that illustrates the sub-files included in the coded symbols obtained in the second iteration in Example \ref{exmp1}}. 
	}
	\label{tab eg1 iter 2}
\end{table*}
\begin{example}
	\label{exmp1}
	Let  $N= K=5, k =1$ and $z=2$. The server stores $5$ files: $\{W^0,W^1,W^2,W^3,W^4\}$, and each file $W^n,n\in \{0,1,2,3,4\}$, is divided into $5$ sub-files: $\{W_0^n,W_1^n,W_2^n,W_3^n,W_4^n\}$. Each cache $C_{\alpha}, \alpha \in \{0,1,2,3,4\},$ is filled with one sub-file $W_{\alpha}^n$ of each file $W^n, n \in  \{0,1,2,3,4\}$,
	and each user $U_{\alpha}, \alpha \in \{0,1,2,3,4\}$, has access to all the caches in $\mathcal{C}_\alpha =\{C_{\alpha},C_{(\alpha+1) \text{ mod } 5}\}$.
	Let the demand vector be ${\bf d}=(d(0),d(1),d(2),d(3),d(4))$.  The parts of the file $W^{d(\alpha)}$ available with the user $U_{\alpha}$ are $2$ consecutive sub-files, i.e.,  $\{W_{\alpha \text{ mod } 5}^{d(\alpha)}, W_{( \alpha+1) \text{ mod } 5}^{d(\alpha)}\}$. Hence,  the user $U_{\alpha}$ should be able to decode all the remaining $3$ sub-files  $\{W_{(\alpha+2) \text{ mod } 5}^{d(\alpha)},W_{(\alpha+3) \text{ mod } 5}^{d(\alpha)},W_{(\alpha+4) \text{ mod } 5}^{d(\alpha)} \}$.

	The coded symbols are transmitted in two iterations. In the first iteration, each user $U_{\alpha}$ retrieves the sub-file $W_{(\alpha+3) \text{ mod } 5}^{d(\alpha)}$ by using the transmissions done in the first iteration, while in the second iteration, each user $U_{\alpha}$ retrieves the sub-files $W_{(\alpha+2) \text{ mod } 5}^{d(\alpha)}$ and $W_{(\alpha+4) \text{ mod } 5}^{d(\alpha)}$ by using the transmissions done in the second iteration.
	
	In the first iteration, each sub-file $W_{(\alpha+3) \text{ mod } 5}^{d(\alpha)}, \forall \alpha \in [0,4],$ is split into $2$ blocks: \\
	$\{ W_{(\alpha+3) \text{ mod } 5,0}^{d(\alpha)},W_{(\alpha+3) \text{ mod } 5,1}^{d(\alpha)}\}$. We obtain five coded symbols, $T_0^1, T_1^1, T_2^1,T_3^1,T_4^1,$ in the first iteration.  Each coded symbol, $T_j^1, j\in [0,4],$ obtained at this iteration, has XOR of blocks of $2$ sub-files where the $2$ sub-files are taken at an interval of $2$ in a cyclic wrap around way, as illustrated in Table \ref{tab eg1 iter 1}. In this table, the columns represent the indices of sub-files and the rows represent the coded symbols. An empty cell means the sub-file (of any file) corresponding to its column is not present in the coded symbol corresponding to its row. An ordered pair $(l,d(n))$ present in any cell, say the cell corresponding to the row of $T_{j}^1$ and to the column $h,$ for some $j,h \in [0,4]$, implies that $W_{h,l}^{d(n)}$ is present in the coded symbol $T_{j}^1$, i.e., $l$ represents the block and $d(n)$ represents the file index demanded by the user $U_n$. The $0^{th}$ block of the first sub-file (indicated by \textcolor{red}{red} color) and the $1^{st}$ block of the second sub-file (indicated by \textcolor{blue}{blue} color) are taken for each coded symbol, as in Table \ref{tab eg1 iter 1}. The following coded symbols are transmitted at this iteration:
{\small
\begin{align*}
T_0^1=& \textcolor{red}{W_{0,0}^{d(2)}} \oplus \textcolor{blue}{W_{2,1}^{d(4)}}, &T_1^1=& \textcolor{red}{W_{1,0}^{d(3)}} \oplus \textcolor{blue}{W_{3,1}^{d(0)}}, \\	T_2^1=& \textcolor{red}{W_{2,0}^{d(4)}} \oplus \textcolor{blue}{W_{4,1}^{d(1)}}, &T_3^1=& \textcolor{red}{W_{3,0}^{d(0)}} \oplus \textcolor{blue}{W_{0,1}^{d(2)}}, \\	T_4^1=& \textcolor{red}{W_{4,0}^{d(1)}} \oplus \textcolor{blue}{W_{1,1}^{d(3)}}.		
\end{align*}}

	Next, in the second iteration, each of the sub-files in $\{W_{(\alpha+2) \text{ mod } 5}^{d(\alpha)},W_{(\alpha+4) \text{ mod } 5}^{d(\alpha)}: \alpha \in [0,4] \} $ is split into $1$ block: 
	$\{W_{(\alpha+2) \text{ mod } 5,0}^{d(\alpha)},W_{(\alpha+4) \text{ mod } 5,0}^{d(\alpha)}: \alpha \in [0,4] \} $ which is basically the sub-file itself. Here, each coded symbol $T_j^2, j\in[0,4],$ consists of $2$ sub-files taken at an interval of $3$ in a cyclic wrap around way, as illustrated in Table \ref{tab eg1 iter 2}. In this table also, like in the previous iteration, an ordered pair $(l,d(n))$ present in any cell, say the cell corresponding to the row of $T_{j}^2$ and to the column $h,$ for some $j,h \in [0,4]$, implies that $W_{h,l}^{d(n)}$ is present in the coded symbol $T_{j}^2$, i.e., the columns represent the indices of sub-files, the rows represent the coded symbols, $l$ represents the block and $d(n)$ represents the file index demanded by the user $U_n$. The following coded symbols are transmitted at the second iteration:
{\small
\begin{align*}
		T_0^2=& \textcolor{red}{W_{0,0}^{d(3)}} \oplus \textcolor{blue}{W_{3,0}^{d(4)}}, &T_1^2=& \textcolor{red}{W_{1,0}^{d(4)}} \oplus \textcolor{blue}{W_{4,0}^{d(0)}}, \\T_2^2=& \textcolor{red}{W_{2,0}^{d(0)}} \oplus \textcolor{blue}{W_{0,0}^{d(1)}}, &T_3^2=& \textcolor{red}{W_{3,0}^{d(1)}} \oplus \textcolor{blue}{W_{1,0}^{d(2)}},\\ T_4^2=& \textcolor{red}{W_{4,0}^{d(2)}} \oplus \textcolor{blue}{W_{2,0}^{d(3)}}.
\end{align*}}
	Now, each user $U_{\alpha}$ needs to recover the demanded file $W^{d(\alpha)}.$  Let us first consider the user $U_0.$  The user $U_0 $ retrieves $W_{3,0}^{d(0)}$ from $T_3^1$ since $W_{0,1}^{d(2)}$ is available at its cache while it retrieves $W_{3,1}^{d(0)}$ from $T_1^1$.  The sub-file $W_2^{d(0)}$- $W_{2,0}^{d(0)} $ is recovered from $T_2^2$ whereas the sub-file $W_4^{d(0)}$- $W_{4,0}^{d(0)} $ is recovered from $T_1^2$. The user $U_0$ has decoded the file $W^{d(0)}$ since it has retrieved all the sub-files corresponding to the file $W^{d(0)}$. Similarly all other users can decode their demanded file.
	In this particular example, the transmission rate achieved using our scheme is $\mathcal{R}_{new}\left ( \frac{1}{5} \right ) = 1.5$ while the rate achieved in \cite{nikhil2020ratememory} is $\mathcal{R}_{RK}\left ( \frac{1}{5} \right )  =1.8.$ 
\end{example}
\begin{table*}[!ht]
	\centering
	\begin{tabular}{ |c|c|c|c|c|c|c|c|c|} 
		\hline
		&{\small  \textit{ $W_0^n$} }&{\small \textit{ $W_1^n$}}&{\small \textit{ $W_2^n$}} &{\small \textit{$W_3^n$}}&{\small  \textit{ $W_4^n$}} &{\small \textit{ $W_5^n$}} &{\small \textit{$W_6^n$}}&{\small  \textit{ $W_7^n$}} \\
		\hline \hline
		{\small $T_{0,1}^1$ }& \textcolor{red}{{\small  $(0,d(3))$}}&  & &\textcolor{blue}{{\small $(0,d(5))$}}&  &	&\textcolor{blue}{{\small $(1,d(0))$}}&   \\ \hline
		{\small $T_{1,1}^1$ }&  &\textcolor{red}{{\small   $(0,d(4))$}}& &&\textcolor{blue}{{\small $(0,d(6))$ } }&	&&\textcolor{blue}{{\small $(1,d(1))$} }  \\ \hline
		{\small $T_{2,1}^1$} &\textcolor{blue}{{\small  $(1,d(2))$}} & &\textcolor{red}{ {\small $(0,d(5))$} }&&  &\textcolor{blue}{{\small $(0,d(7))$}}	&&   \\ \hline
		{\small $T_{3,1}^1$} &  &\textcolor{blue}{{\small $(1,d(3))$}} & &\textcolor{red}{{\small $(0,d(6))$}}&  &	&\textcolor{blue}{{\small $(0,d(0))$}}& \\ \hline
		{\small $T_{4,1}^1$} &  & &\textcolor{blue}{{\small  $(1,d(4))$}}&& \textcolor{red}{{\small $(0,d(7))$}}&	&& \textcolor{blue}{{\small $(0,d(1))$ }}\\ \hline
		{\small $T_{5,1}^1$} & \textcolor{blue}{{\small $(0,d(2))$ }}& & &\textcolor{blue}{{\small $(1,d(5))$}}&  &\textcolor{red}{{\small $(0,d(0))$}}	&& \\ \hline
		{\small $T_{6,1}^1$} &  &\textcolor{blue}{{\small $(0,d(3))$ }}& && \textcolor{blue}{{\small $(1,d(6))$}} &	&\textcolor{red}{{\small $(0,d(1))$}}&   \\ \hline
		{\small $T_{7,1}^1$} &  & &\textcolor{blue}{{\small $(0,d(4))$}} &&  &	\textcolor{blue}{{\small $(1,d(7))$}}&&\textcolor{red}{{\small $(0,d(2))$ }} \\ \hline
	\end{tabular}
	\caption{Table that illustrates the sub-files included in the coded symbols $T_{j,1}^1, j \in [0,7]$, obtained in the first iteration in Example \ref{exmp2}. 
	 }
	\label{tab eg2 iter 1 part 1}
\end{table*}
\begin{table*}[!ht]
	\centering
	\begin{tabular}{  |c|c|c|c|c|c|c|c|c|} 
		\hline
		&{\small  \textit{ $W_0^n$}} &{\small \textit{ $W_1^n$}}&{\small \textit{ $W_2^n$}} &{\small \textit{$W_3^n$}}&{\small  \textit{ $W_4^n$}}&{\small \textit{ $W_5^n$}} &{\small \textit{$W_6^n$}}&{\small  \textit{ $W_7^n$}} \\
		\hline \hline
	 {\small $T_{0,2}^1$} & \textcolor{red}{{\small  $(1,d(3))$}}&  & &\textcolor{red}{{\small $(2,d(6))$}}&&&\textcolor{blue}{{\small $(2,d(0))$}}&  \\ \hline
		{\small $T_{1,2}^1$ }&  &\textcolor{red}{{\small $(1,d(4))$}}  & &&\textcolor{red}{{\small $(2,d(7))$}}&&&\textcolor{blue}{{\small $(2,d(1))$ }} \\ \hline
	{\small $T_{2,2}^1$} & \textcolor{blue}{{\small $(2,d(2))$ }} & &\textcolor{red}{{\small  $(1,d(5))$ }}&&&\textcolor{red}{{\small $(2,d(0))$}}&& \\ \hline
		{\small $T_{3,2}^1$} &   &\textcolor{blue}{{\small $(2,d(3))$} }&  &\textcolor{red}{{\small $(1,d(6))$}}&&&\textcolor{red}{{\small $(2,d(1))$}}& \\ \hline
	{\small $T_{4,2}^1$ }&   & & \textcolor{blue}{{\small $(2,d(4))$} }&&\textcolor{red}{{\small $(1,d(7))$}}&&&\textcolor{red}{{\small $(2,d(2))$ }}\\ \hline
		{\small $T_{5,2}^1$ }& \textcolor{red}{{\small $(2,d(3))$}}  & & &\textcolor{blue}{{\small $(2,d(5))$} }&&\textcolor{red}{{\small $(1,d(0))$}}&& \\ \hline
		{\small $T_{6,2}^1$} &   &\textcolor{red}{{\small $(2,d(4))$}} & &&\textcolor{blue}{{\small $(2,d(6))$} }&&\textcolor{red}{{\small $(1,d(1))$}}& \\ \hline
		{\small $T_{7,2}^1$} &   & &\textcolor{red}{{\small $(2,d(5))$ }}&& &\textcolor{blue}{{\small $(2,d(7))$}}&& \textcolor{red}{{\small $(1,d(2))$}}\\ \hline
	\end{tabular}
	\caption{Table that illustrates the sub-files included in the coded symbols $T_{j,2}^1, j \in [0,7]$, obtained in the first iteration in Example \ref{exmp2}. 
	}
	\label{tab eg2 iter 1 part 2}
\end{table*}
%
%
%
%
%
%

\begin{table*}[!ht]
\centering
\begin{tabular}{ |c||c|c|c|c|c|c|c|c|} 
\hline
&{\small  \textit{ $W_0^n$}} &{\small \textit{ $W_1^n$}}&{\small \textit{ $W_2^n$}} &{\small \textit{$W_3^n$}}&{\small  \textit{ $W_4^n$}}&{\small \textit{ $W_5^n$}} &{\small \textit{$W_6^n$}}&{\small  \textit{ $W_7$}} \\
\hline \hline
{\small $T_{0}^2$} &  \textcolor{red}{{\small $(0,d(4))$}}&  & &&\textcolor{blue}{{\small $(0,d(5))$}}&&&  \\ \hline
{\small $T_{1}^2$} & & \textcolor{red}{{\small  $(0,d(5))$}}& &&&\textcolor{blue}{{\small $(0,d(6))$}}&&  \\ \hline
{\small $T_{2}^2$} & &  &\textcolor{red}{{\small $(0,d(6))$}}&&&&\textcolor{blue}{{\small $(0,d(7))$}}&  \\ \hline
{\small $T_{3}^2$ }&   &  &&\textcolor{red}{{\small $(0,d(7))$}}&&&&\textcolor{blue}{{\small $(0,d(0))$} } \\ \hline
{\small $T_{4}^2$} & \textcolor{blue}{{\small  $(0,d(1))$} }  &  &&&\textcolor{red}{{\small $(0,d(1))$}}&&&\\ \hline
{\small $T_{5}^2$ }&   &\textcolor{blue}{{\small $(0,d(2))$}} & & &&\textcolor{red}{{\small $(0,d(2))$}}&& \\ \hline
{\small $T_{6}^2$ }&   & &\textcolor{blue}{{\small $(0,d(3))$ }}&& &&\textcolor{red}{{\small $(0,d(3))$}}& \\ \hline
{\small $T_{7}^2$} &   & & &\textcolor{blue}{{\small $(0,d(4))$}}& &&&\textcolor{red}{ {\small $(0,d(4))$}}\\ \hline
\end{tabular}
\caption{ Table that illustrates the sub-files included in the coded symbols $T_{j}^2, j \in [0,7]$ obtained in the second iteration in Example \ref{exmp2}.  
}
\label{tab eg2 iter 2}
\end{table*}
\begin{example}
	\label{exmp2}
	Let $N=K=8, k =1$ and $z=4.$  The server stores $8$ files: $\{W^0,W^1, \ldots,W^7\}$ and each file $W^n,n\in [0,7]$, is divided into $8$ sub-files:  $\{W_0^n,W_1^n, \ldots ,W_7^n\}$. Each cache $C_{\alpha}, \alpha \in [0,7],$ is filled with one sub-file $W_{\alpha}^n$ of each file $W^n, n \in  [0,7]$ and
	each user $U_{\alpha}, \alpha \in [0,7]$, has access to all the caches in $\mathcal{C}_\alpha =\{C_{\alpha},C_{(\alpha+1) \text{ mod } 8},C_{(\alpha+2) \text{ mod } 8},C_{(\alpha+3) \text{ mod } 8}\}$. Let the demand vector be ${\bf{d}}=(d(0),d(1),d(2),d(3),d(4),d(5),d(6),d(7))$. The parts of the file $W^{d(\alpha)}$ available with the user $U_{\alpha}$ are $4$ consecutive sub-files, i.e.,  $\{W_{\alpha \text{ mod } 8}^{d(\alpha)}, W_{( \alpha+1) \text{ mod } 8}^{d(\alpha)} ,  W_{( \alpha+2 ) \text{ mod } 8}^{d(\alpha)} , W_{( \alpha+3 ) \text{ mod } 8}^{d(\alpha)}\}$. Hence,  the user $U_{\alpha}$ should be able to decode all the remaining $4$ sub-files, i.e., \\ $\{W_{(\alpha+4) \text{ mod } 8}^{d(\alpha)},W_{(\alpha+5) \text{ mod } 8}^{d(\alpha)},W_{(\alpha+6) \text{ mod } 8}^{d(\alpha)},W_{(\alpha+7) \text{ mod } 8}^{d(\alpha)} \}$.
	
	The coded symbols are transmitted in two iterations. 
	In the first iteration, each user $U_{\alpha}$ retrieves the sub-files $W_{(\alpha+5) \text{ mod } 8}^{d(\alpha)}$ and $W_{(\alpha+6) \text{ mod } 8}^{d(\alpha)}$ by using the transmissions done at the first iteration while  in the second iteration, each user $U_{\alpha}$ retrieves the sub-files $W_{(\alpha+4) \text{ mod } 8}^{d(\alpha)}$ and $W_{(\alpha+7) \text{ mod } 8}^{d(\alpha)}$ by using the transmissions done at the second iteration.
	
	In the first iteration, each of the sub-files in $\{W_{(\alpha+5) \text{ mod } 8}^{d(\alpha)},W_{(\alpha+6) \text{ mod } 8}^{d(\alpha)}: \alpha \in [0,7] \} $ is split into $3$ blocks: 
	$\{W_{(\alpha+5) \text{ mod } 8,l}^{d(\alpha)},W_{(\alpha+6) \text{ mod } 8,l}^{d(\alpha)}: \alpha \in [0,7], l \in [0,2] \} $. Here, each coded symbol consists of $3$ sub-files taken at an interval of $3$ in a cyclic wrap around way, as illustrated in Tables \ref{tab eg2 iter 1 part 1} and \ref{tab eg2 iter 1 part 2}.  In both the tables, the columns represent the indices of sub-files and the rows represent the coded symbols. An ordered pair $(l,d(n))$ present in any cell, say the cell corresponding to the row of $T_{j,i}^1$ and to the column $h,$ for some $j,h \in [0,7], i\in \{1,2\}$, implies that $W_{h,l}^{d(n)}$ is present in the coded symbol $T_{j,i}^1$. The $0^{th}$ block of the first and the second sub-files and the $1^{st}$ block of the third sub-file are taken in each coded symbol $T_{j,1}^{1}$ as in Table \ref{tab eg2 iter 1 part 1} while in each coded symbol $T_{j,2}^{1}$,  the $1^{st}$ block of the first sub-file, the $2^{nd}$ block of the second and the third sub-files are taken as in Table \ref{tab eg2 iter 1 part 2}.   The following coded symbols are transmitted at the first iteration:
{\small
\begin{align*}
T_{0,1}^1 =& \textcolor{red}{W_{0,0}^{d(3)}} \oplus \textcolor{blue}{W_{3,0}^{d(5)}} \oplus \textcolor{blue}{W_{6,1}^{d(0)}} & T_{0,2}^1 =& \textcolor{red}{W_{0,1}^{d(3)}} \oplus \textcolor{red}{W_{3,2}^{d(6)}} \oplus \textcolor{blue}{W_{6,2}^{d(0)}} \\ T_{1,1}^1 =& \textcolor{red}{W_{1,0}^{d(4)}} \oplus \textcolor{blue}{W_{4,0}^{d(6)}} \oplus \textcolor{blue}{W_{7,1}^{d(1)}}&
T_{1,2}^1=& \textcolor{red}{W_{1,1}^{d(4)}} \oplus \textcolor{red}{W_{4,2}^{d(7)}} \oplus \textcolor{blue}{W_{7,2}^{d(1)}} \\ T_{2,1}^1 =& \textcolor{red}{W_{2,0}^{d(5)}} \oplus \textcolor{blue}{W_{5,0}^{d(7)}} \oplus \textcolor{blue}{W_{0,1}^{d(2)}} & 	T_{2,2}^1 =& \textcolor{red}{W_{2,1}^{d(5)}} \oplus \textcolor{red}{W_{5,2}^{d(0)}} \oplus \textcolor{blue}{W_{0,2}^{d(2)}}	\\	
T_{3,1}^1 =& \textcolor{red}{W_{3,0}^{d(6)}} \oplus \textcolor{blue}{W_{6,0}^{d(0)}} \oplus \textcolor{blue}{W_{1,1}^{d(3)}}& 	T_{3,2}^1 =& \textcolor{red}{W_{3,1}^{d(6)}} \oplus \textcolor{red}{W_{6,2}^{d(1)}} \oplus \textcolor{blue}{W_{1,2}^{d(3)}} \\ T_{4,1}^1 =& \textcolor{red}{W_{4,0}^{d(7)}} \oplus \textcolor{blue}{W_{7,0}^{d(1)}} \oplus \textcolor{blue}{W_{2,1}^{d(4)}} &
T_{4,2}^1 =& \textcolor{red}{W_{4,1}^{d(7)}} \oplus \textcolor{red}{W_{7,2}^{d(2)}} \oplus \textcolor{blue}{W_{2,2}^{d(4)}} \\ T_{5,1}^1 =& \textcolor{red}{W_{5,0}^{d(0)}} \oplus \textcolor{blue}{W_{0,0}^{d(2)}} \oplus \textcolor{blue}{W_{3,1}^{d(5)}} &	 T_{5,2}^1 =& \textcolor{red}{W_{5,1}^{d(0)}} \oplus \textcolor{red}{W_{0,2}^{d(3)}} \oplus \textcolor{blue}{W_{3,2}^{d(5)}}\\
T_{6,1}^1 =& \textcolor{red}{W_{6,0}^{d(1)}} \oplus \textcolor{blue}{W_{1,0}^{d(3)}}\oplus \textcolor{blue}{W_{4,1}^{d(6)}} &T_{6,2}^1 =& \textcolor{red}{W_{6,1}^{d(1)}} \oplus \textcolor{red}{W_{1,2}^{d(4)}}\oplus \textcolor{blue}{W_{4,2}^{d(6)}} \\ T_{7,1}^1 =& \textcolor{red}{W_{7,0}^{d(2)}} \oplus \textcolor{blue}{W_{2,0}^{d(4)}} \oplus \textcolor{blue}{W_{5,1}^{d(7)}} &
	T_{7,2}^1 =& \textcolor{red}{W_{7,1}^{d(2)}} \oplus \textcolor{red}{W_{2,2}^{d(5)}} \oplus \textcolor{blue}{W_{5,2}^{d(7)}} & & 
\end{align*}}	
	In the second iteration, each of the sub-files in $\{W_{(\alpha+4) \text{ mod } 8}^{d(\alpha)},W_{(\alpha+7) \text{ mod } 8}^{d(\alpha)}: \alpha \in [0,7] \} $ is split into $1$ block: 
	$\{W_{(\alpha+4) \text{ mod } 8,0}^{d(\alpha)},W_{(\alpha+7) \text{ mod } 8,0}^{d(\alpha)}: \alpha \in [0,7] \} $, which is basically the sub-file itself. Here, each coded symbol consists of $2$ sub-files taken at an interval of $4$ in a cyclic wrap around way as illustrated in Table \ref{tab eg2 iter 2}. In this table also, like in the previous iteration, an ordered pair $(l,d(n))$ present in any cell, say the cell corresponding to the row of $T_{j}^2$ and to the column $h,$ for some $j,h \in [0,4]$, implies that $W_{h,l}^{d(n)}$ is present in the coded symbol $T_{j}^2$, i.e., the columns represent the indices of sub-files, the rows represent the coded symbols, $l$ represents the block and $d(n)$ represents the file index demanded by the user $U_n$. The following coded symbols are transmitted during the second iteration:
	{\small
	\begin{align*}
		T_0^2=& \textcolor{red}{W_{0,0}^{d(4)}} \oplus \textcolor{blue}{W_{4,0}^{d(5)}}, & T_1^2=& \textcolor{red}{W_{1,0}^{d(5)}} \oplus \textcolor{blue}{W_{5,0}^{d(6)}}, \\	T_2^2=& \textcolor{red}{W_{2,0}^{d(6)}} \oplus \textcolor{blue}{W_{6,0}^{d(7)}}, & 	T_3^2=& \textcolor{red}{W_{3,0}^{d(7)}} \oplus \textcolor{blue}{W_{7,0}^{d(0)}}, \\ T_4^2=& \textcolor{red}{W_{4,0}^{d(0)}} \oplus \textcolor{blue}{W_{0,0}^{d(1)}},   &	T_5^2=& \textcolor{red}{W_{5,0}^{d(1)}} \oplus \textcolor{blue}{W_{1,0}^{d(2)}} \\	 T_6^2=& \textcolor{red}{W_{6,0}^{d(2)}} \oplus \textcolor{blue}{W_{2,0}^{d(3)}}, &
		T_7^2=& \textcolor{red}{W_{7,0}^{d(3)}} \oplus \textcolor{blue}{W_{3,0}^{d(4)}}.
		\end{align*}}
	Now, each user $U_{\alpha}$ needs to recover the demanded file $W^{d(\alpha)}$ from the above transmissions. Let us first consider the user $U_0.$  The user $U_0 $ retrieves $W_{5,0}^{d(0)}$ from $T_{5,1}^1$ while it retrieves $W_{5,1}^{d(0)}$ and $W_{5,2}^{d(0)}$ from $T_{5,2}^1$ and $T_{2,2}^1$ respectively. Similarly, the user $U_0 $ retrieves $W_{6,0}^{d(0)}$, $W_{6,1}^{d(0)}$ and $W_{6,2}^{d(0)}$ from $T_{3,1}^1,T_{0,1}^1$ and $T_{0,2}^1$ respectively. The sub-file $W_4^{d(0)}$- $W_{4,0}^{d(0)} $ is recovered from $T_4^2$ whereas the sub-file $W_7^{d(0)}$- $W_{7,0}^{d(0)} $ is recovered from $T_3^2$. The user $U_0$ has decoded the file $W^{d(0)}$ since it has retrieved all the sub-files corresponding to the file $W^{d(0)}$. Similarly all other users can decode their demanded file. 
	In this example, the transmission rate achieved using our scheme is $\mathcal{R}_{new}\left ( \frac{1}{8} \right )  = \frac{5}{3}$ while $\mathcal{R}_{RK}\left ( \frac{1}{8} \right )  =2$ in  \cite{nikhil2020ratememory}.
\end{example}
\begin{example}
\label{exmp3}
	Let us take an example when $k=3.$ Let $N=K=10, k =3$ and $z=2.$  The server stores $10$ files: $\{W^0,W^1, \ldots,W^9\}$ and each file $W^n,n\in [0,9]$, is divided into $10$ sub-files:  $\{W_0^n,W_1^n, \ldots ,W_9^n\}$. Each cache $C_{\alpha}, \alpha \in [0,9],$ is filled as 	$M_{\alpha}= \{W_{(3 \alpha +j) \text{ mod } K}^n:  j \in [0,2],   n \in [0,9]\}$, i.e.,
{\small
\begin{align*}
M_0 =& \{W_0^n, W_1^n,W_2^n : n \in [0,9]\}\\	
			M_1 =& \{W_3^n,W_4^n,W_5^n : n \in [0,9]\}\\	M_2 =& \{W_6^n,W_7^n,W_8^n : n \in [0,9]\}\\
M_3 =& \{W_9^n, W_0^n,W_1^n : n \in [0,9]\} \\	M_4 =& \{W_2^n,W_3^n,W_4^n : n \in [0,9]\} \\     	M_5 =& \{W_5^n,W_6^n,W_7^n : n \in [0,9]\} \\ 
M_6 =& \{W_8^n, W_9^n, W_0^n : n \in [0,9]\}\\  	M_7 =& \{W_1^n,W_2^n,W_3^n : n \in [0,9]\}\\	  M_8 =& \{W_4^n,W_5^n,W_6^n : n \in [0,9]\} \\ 	 M_9 =& \{W_7^n,W_8^n,W_9^n : n \in [0,9]\} 
\end{align*}}
	Each user $U_{\alpha}, \alpha \in [0,8]$, has access to all the caches in $\mathcal{C}_\alpha =\{C_{\alpha},C_{(\alpha+1) \text{ mod } 10}\}$. Let the demand vector be ${\bf{d}}=(d(0),d(1),d(2),d(3),d(4),d(5),d(6),d(7),d(8),d(9))$. 	The parts of the file $W^{d(\alpha)}$ available with the user $U_{\alpha}$ are $6$ consecutive sub-files, i.e.,  $\{W_{(3 \alpha) \text{ mod } 10}^{d(\alpha)}, W_{(3 \alpha+1) \text{ mod } 10}^{d(\alpha)} , \ldots , W_{(3 \alpha+5 ) \text{ mod } 10}^{d(\alpha)}\}$. Hence,  the user $U_{\alpha}$ should be able to decode all the remaining $4$ sub-files, i.e., \\ $\{W_{(3\alpha+6) \text{ mod } 10}^{d(\alpha)},W_{(3\alpha+7) \text{ mod } 10}^{d(\alpha)},W_{(3\alpha+8) \text{ mod } 10}^{d(\alpha)},W_{(3\alpha+9) \text{ mod } 10}^{d(\alpha)} \}$.
	
		The coded symbols are transmitted in two iterations. In the first iteration, each user $U_{\alpha}$ retrieves the sub-files $W_{(3\alpha+7) \text{ mod } 10}^{d(\alpha)}$ and $W_{(3\alpha+8) \text{ mod } 10}^{d(\alpha)}$ by using the transmissions done at the first iteration while  in the second iteration, each user $U_{\alpha}$ retrieves the sub-files $W_{(3\alpha+6) \text{ mod } 10}^{d(\alpha)}$ and $W_{(3\alpha+9) \text{ mod } 10}^{d(\alpha)}$ by using the transmissions done at the second iteration.
	
	In the first iteration, each of the sub-files in $\{W_{(3\alpha+7) \text{ mod } 10}^{d(\alpha)},W_{(3\alpha+8) \text{ mod } 10}^{d(\alpha)}: \alpha \in [0,9] \} $ is split into $3$ blocks: 
	$\{W_{(3\alpha+7) \text{ mod } 10,l}^{d(\alpha)},W_{(3\alpha+8) \text{ mod } 10,l}^{d(\alpha)}: \alpha \in [0,9], l \in [0,2] \} $. Here, each coded symbol is XOR of blocks of $3$ sub-files taken at an interval of $3$ in a cyclic wrap around way.  The $0^{th}$ block of the first and the second sub-files and the $1^{st}$ block of the third sub-file are taken in each coded symbol $T_{j,1}^{1}$ while in each coded symbol $T_{j,2}^{1}$,  the $1^{st}$ block of the first sub-file, the $2^{nd}$ block of the second and the third sub-files are taken.  The following coded symbols are transmitted during the first iteration:	
{\small
\begin{align*}
T_{0,1}^1=& \textcolor{red}{W_{0,0}^{d(1)}} \oplus \textcolor{blue}{W_{3,0}^{d(5)}} \oplus \textcolor{blue}{W_{6,1}^{d(6)}},  & T_{0,2}^1=& \textcolor{red}{W_{0,1}^{d(1)}} \oplus \textcolor{red}{W_{3,2}^{d(2)}} \oplus \textcolor{blue}{W_{6,2}^{d(6)}}, \\   T_{1,1}^1=& \textcolor{red}{W_{1,0}^{d(8)}} \oplus \textcolor{blue}{W_{4,0}^{d(2)} }\oplus \textcolor{blue}{W_{7,1}^{d(3)}} , &
T_{1,2}^1=& \textcolor{red}{W_{1,1}^{d(8)}} \oplus \textcolor{red}{W_{4,2}^{d(9)}} \oplus \textcolor{blue}{W_{7,2}^{d(3)}} , \\  T_{2,1}^1=& \textcolor{red}{W_{2,0}^{d(5)}} \oplus \textcolor{blue}{W_{5,0}^{d(9)}} \oplus \textcolor{blue}{W_{8,1}^{d(0)}}, & T_{2,2}^1=& \textcolor{red}{W_{2,1}^{d(5)}} \oplus \textcolor{red}{W_{5,2}^{d(6)}} \oplus \textcolor{blue}{W_{8,2}^{d(0)}}, \\  
T_{3,1}^1=& \textcolor{red}{W_{3,0}^{d(2)}} \oplus \textcolor{blue}{W_{6,0}^{d(6)}} \oplus \textcolor{blue}{W_{9,1}^{d(7)}},   & T_{3,2}^1=& \textcolor{red}{W_{3,1}^{d(2)}} \oplus \textcolor{red}{W_{6,2}^{d(3)}} \oplus \textcolor{blue}{W_{9,2}^{d(7)}}, \\	T_{4,1}^1=& \textcolor{red}{W_{4,0}^{d(9)}} \oplus \textcolor{blue}{W_{7,0}^{d(3)}} \oplus \textcolor{blue}{W_{0,1}^{d(4)}},&
T_{4,2}^1=& \textcolor{red}{W_{4,1}^{d(9)}} \oplus \textcolor{red}{W_{7,2}^{d(0)}} \oplus \textcolor{blue}{W_{0,2}^{d(4)}}, \\  T_{5,1}^1=& \textcolor{red}{W_{5,0}^{d(6)}} \oplus \textcolor{blue}{W_{8,0}^{d(0)}} \oplus \textcolor{blue}{W_{1,1}^{d(1)}} ,   & T_{5,2}^1=& \textcolor{red}{W_{5,1}^{d(6)}} \oplus \textcolor{red}{W_{8,2}^{d(7)}} \oplus \textcolor{blue}{W_{1,2}^{d(1)} }, \\
T_{6,1}^1=& \textcolor{red}{W_{6,0}^{d(3)}} \oplus \textcolor{blue}{W_{9,0}^{d(7)}} \oplus \textcolor{blue}{W_{2,1}^{d(8)}},  & T_{6,2}^1=& \textcolor{red}{W_{6,1}^{d(3)}} \oplus \textcolor{red}{W_{9,2}^{d(4)}} \oplus \textcolor{blue}{W_{2,2}^{d(8)}}, \\  T_{7,1}^1=& \textcolor{red}{W_{7,0}^{d(0)}} \oplus \textcolor{blue}{W_{0,0}^{d(4)}} \oplus \textcolor{blue}{W_{3,1}^{d(5)}},&
 T_{7,2}^1=& \textcolor{red}{W_{7,1}^{d(0)}} \oplus \textcolor{red}{W_{0,2}^{d(1)}} \oplus \textcolor{blue}{W_{3,2}^{d(5)}}, \\  T_{8,1}^1=& \textcolor{red}{W_{8,0}^{d(7)}} \oplus \textcolor{blue}{W_{1,0}^{d(1)}} \oplus \textcolor{blue}{W_{4,1}^{d(2)}},  &  T_{8,2}^1=& \textcolor{red}{W_{8,1}^{d(7)}} \oplus \textcolor{red}{W_{1,2}^{d(8)}} \oplus \textcolor{blue}{W_{4,2}^{d(2)}},\\
T_{9,1}^1=& \textcolor{red}{W_{9,0}^{d(4)}} \oplus \textcolor{blue}{W_{2,0}^{d(8)}} \oplus \textcolor{blue}{W_{5,1}^{d(9)}},  &  T_{9,2}^1=& \textcolor{red}{W_{9,1}^{d(4)}} \oplus \textcolor{red}{W_{2,2}^{d(5)}} \oplus \textcolor{blue}{W_{5,2}^{d(9)}}. & 
\end{align*}}
	In the second iteration, each of the sub-files in $\{W_{(3\alpha+6) \text{ mod } 10}^{d(\alpha)},W_{(3\alpha+9) \text{ mod } 10}^{d(\alpha)}: \alpha \in [0,9] \} $ is split into $3$ blocks: 
$\{W_{(3\alpha+6) \text{ mod } 10,l}^{d(\alpha)},W_{(3\alpha+9) \text{ mod } 10,l}^{d(\alpha)}: \alpha \in [0,9], l \in [0,2] \} $. Here, each coded symbol consists of $3$ sub-files taken at an interval of $4$ in a cyclic wrap around way. The $0^{th}$ block of the first and the second sub-files and the $1^{st}$ block of the third sub-file are taken in each coded symbol $T_{j,1}^{1}$ while in each coded symbol $T_{j,2}^{1}$,  the $1^{st}$ block of the first sub-file, the $2^{nd}$ block of the second and the third sub-files are taken.   The following coded symbols are transmitted during the second iteration:
{\small
\begin{align*}
T_{0,1}^2=& \textcolor{red}{W_{0,0}^{d(8)}} \oplus \textcolor{blue}{W_{4,0}^{d(5)}} \oplus \textcolor{blue}{W_{8,1}^{d(3)}},  & T_{0,2}^2=& \textcolor{red}{W_{0,1}^{d(8)}} \oplus \textcolor{red}{W_{4,2}^{d(6)}} \oplus \textcolor{blue}{W_{8,2}^{d(3)}}, \\  T_{1,1}^2=& \textcolor{red}{W_{1,0}^{d(5)}} \oplus \textcolor{blue}{W_{5,0}^{d(2)}} \oplus \textcolor{blue}{W_{9,1}^{d(0)}}, &
T_{1,2}^2=& \textcolor{red}{W_{1,1}^{d(5)}} \oplus \textcolor{red}{W_{5,2}^{d(3)}} \oplus \textcolor{blue}{W_{9,2}^{d(0)}}, \\  T_{2,1}^2=& \textcolor{red}{W_{2,0}^{d(2)}} \oplus \textcolor{blue}{W_{6,0}^{d(9)}} \oplus \textcolor{blue}{W_{0,1}^{d(7)}}, & T_{2,2}^2=& \textcolor{red}{W_{2,1}^{d(2)}} \oplus \textcolor{red}{W_{6,2}^{d(0)}} \oplus \textcolor{blue}{W_{0,2}^{d(7)}}, \\  
T_{3,1}^2=& \textcolor{red}{W_{3,0}^{d(9)}} \oplus \textcolor{blue}{W_{7,0}^{d(6)}} \oplus \textcolor{blue}{W_{1,1}^{d(4)}},   & T_{3,2}^2=& \textcolor{red}{W_{3,1}^{d(9)}} \oplus \textcolor{red}{W_{7,2}^{d(7)}} \oplus \textcolor{blue}{W_{1,2}^{d(4)}}, \\	T_{4,1}^2=& \textcolor{red}{W_{4,0}^{d(6)} }\oplus \textcolor{blue}{W_{8,0}^{d(3)}} \oplus \textcolor{blue}{W_{2,1}^{d(1)}},&
T_{4,2}^2=& \textcolor{red}{W_{4,1}^{d(6)}} \oplus \textcolor{red}{W_{8,2}^{d(4)}} \oplus \textcolor{blue}{W_{2,2}^{d(1)}}, \\  T_{5,1}^2=& \textcolor{red}{W_{5,0}^{d(3)}} \oplus \textcolor{blue}{W_{9,0}^{d(0)}} \oplus \textcolor{blue}{W_{3,1}^{d(8)}} , & T_{5,2}^2=& \textcolor{red}{W_{5,1}^{d(3)}} \oplus \textcolor{red}{W_{9,2}^{d(1)}} \oplus \textcolor{blue}{W_{3,2}^{d(8)} }, \\ 	
T_{6,1}^2=& \textcolor{red}{W_{6,0}^{d(0)}} \oplus \textcolor{blue}{W_{0,0}^{d(7)}} \oplus \textcolor{blue}{W_{4,1}^{d(5)}},  & T_{6,2}^2=& \textcolor{red}{W_{6,1}^{d(0)}} \oplus \textcolor{red}{W_{0,2}^{d(8)}} \oplus \textcolor{blue}{W_{4,2}^{d(5)}}, \\  T_{7,1}^2=& \textcolor{red}{W_{7,0}^{d(7)}} \oplus \textcolor{blue}{W_{1,0}^{d(4)}} \oplus \textcolor{blue}{W_{5,1}^{d(2)}},&
T_{7,2}^2=& \textcolor{red}{W_{7,1}^{d(7)}} \oplus \textcolor{red}{W_{1,2}^{d(5)}} \oplus \textcolor{blue}{W_{5,2}^{d(2)}}, \\  T_{8,1}^2=& \textcolor{red}{W_{8,0}^{d(4)}} \oplus \textcolor{blue}{W_{2,0}^{d(1)}} \oplus \textcolor{blue}{W_{6,1}^{d(9)}},  &  T_{8,2}^2=& \textcolor{red}{W_{8,1}^{d(4)}} \oplus \textcolor{red}{W_{2,2}^{d(2)}} \oplus \textcolor{blue}{W_{6,2}^{d(9)}}, \\  
T_{9,1}^2=& \textcolor{red}{W_{9,0}^{d(1)}} \oplus \textcolor{blue}{W_{3,0}^{d(8)}} \oplus \textcolor{blue}{W_{7,1}^{d(6)}},  &  T_{9,2}^2=& \textcolor{red}{W_{9,1}^{d(1)}} \oplus \textcolor{red}{W_{3,2}^{d(9)}} \oplus \textcolor{blue}{W_{7,2}^{d(6)}}. &
\end{align*}}
	Now, each user $U_{\alpha}$ needs to recover the demanded file $W^{d(\alpha)}$ from the above transmissions. Let us first consider the user $U_0.$  The user $U_0 $ retrieves $W_{7,0}^{d(0)}$ from $T_{7,1}^1$ while it retrieves $W_{7,1}^{d(0)}$ and $W_{7,2}^{d(0)}$ from $T_{7,2}^1$ and $T_{4,2}^1$ respectively. Similarly, the user $U_0 $ retrieves $W_{8,0}^{d(0)}$, $W_{8,1}^{d(0)}$ and $W_{8,2}^{d(0)}$ from $T_{5,1}^1, T_{2,1}^1$ and $T_{2,2}^1$ respectively. The user $U_0 $ gets $W_{6,0}^{d(0)}$, $W_{6,1}^{d(0)}$ and $W_{6,2}^{d(0)}$ from $T_{6,1}^2, T_{6,2}^2$ and $T_{2,2}^2$ respectively while it gets $W_{9,0}^{d(0)}$, $W_{9,1}^{d(0)}$ and $W_{9,2}^{d(0)}$ from $T_{5,1}^2, T_{1,1}^2$ and $T_{1,2}^2$ respectively. The user $U_0$ has decoded the file $W^{d(0)}$ since it has retrieved all the sub-files corresponding to the file $W^{d(0)}$. Similarly all other users can decode their demanded file. 
	In this example, the transmission rate achieved using our scheme is $\mathcal{R}_{new}\left ( \frac{3}{10} \right )  = \frac{4}{3}$ while $\mathcal{R}_{RK}\left ( \frac{3}{10} \right )  =1.6$ in  \cite{nikhil2020ratememory}.
\end{example}
\subsection{On the lower convex envelope of the achievable rates}
In Fig. 3(b), transmission rate vs $\gamma$  plot is obtained for $K=25,z =3$ for all the points mentioned in Theorem \ref{thm1} along with the lower convex envelop of all the points mentioned in Theorem \ref{thm1} omitting $\gamma=\frac{2}{25},\frac{3}{25}$ since the line joining the points $\gamma=\frac{1}{25}$ and $\gamma=\frac{4}{25}$ falls below those two points. In general we conjecture that points corresponding to $kz<(\frac{K-1}{2})$ except for $kz=z$  will have this characteristic. This may be mainly due to taking the ceiling operation of certain values in the rate expression. However, irrespective of this nature of some points all the points mentioned in Theorem \ref{thm1} fall below the curve obtained for the scheme in \cite{nikhil2020ratememory}.
\subsection{Sub-packetization Level}
For any $\gamma \in \left \{\frac{k}{K}: \gcd(k,K)=1,  k\in \left [1,K \right ]\right \}$, Algorithm \ref{algo2} is used to derive the transmissions done by the server. In  Algorithm \ref{algo2}, the worst case sub-packetization level is when $p$ is odd. In that case each sub-file is divided further into $p$ parts, where $p=\lceil \frac{kz}{r} \rceil+1$. 
To maximize $p,$ the maximum value that $kz$ can take needs to be chosen, which is when $kz=K-1$. Hence the maximum value that $p$ can take is $K$, i.e., each sub-file is further divided into at most $K$ parts in the worst case scenario. Thus, the worst case sub-packetization level in our scheme is $K^2.$ The sub-packetization level required for the scheme proposed in \cite{nikhil2020ratememory} is ${K-kz+k-1 \choose k-1} \frac{K}{k}$ while the sub-packetization level required for the scheme in \cite{Caire2020noveltransformation} is ${K-kz+k \choose k} K$. So, the sub-packetization level required for our scheme is less compared to both these schemes. The sub-packetization level required for the scheme proposed in \cite{parinello2019multiacessgains} is $K$, which is less than the sub-packetization level required for our scheme.
\section{Discussion}
\label{sec4}
In this work, we have presented a placement and delivery scheme for multi-access coded caching problem, under the restriction of uncoded placement, with each cache having a normalized capacity of $ \gamma$, where $\gamma \in \left \{\frac{k}{K}: \gcd(k,K)=1,  k\in \left [1,K \right ]\right \}$. We have shown that our work is  a generalization of one of the cases considered in \cite{parinello2019multiacessgains}. We have also proved that our scheme outperforms that in \cite{nikhil2020ratememory} for the cases under consideration. Here, we assume that each user has access to same number of caches and each cache is of same capacity which need not be true in practical scenarios. Hence, it is a good direction to work on when the cache sizes are heterogeneous and each user has access to random number of users.

    In \cite{KMR2020CRD}, the authors have identified a special class of resolvable designs called cross resolvable designs which led to multi-access coded caching schemes but with the number of users being different from the number of caches. 

 
\begin{center}
	APPENDIX A
	
	Proof of Correctness of the Delivery Scheme
\end{center}
Recall that, for each user $U_{\alpha} , \alpha\in [0,K-1]$, the  accessible cache content is $\{W_{(k \alpha + i) \text{ mod } K}^{n} : i=0,1,\dots,kz-1,  n \in [0,N-1]\}$. Hence, in order to retrieve the file $W^{d(\alpha)}$,  the user $U_{\alpha}$ needs to decode all the remaining $K-kz$ sub-files, i.e.,  $\{W_{k \alpha+kz+i \text{ mod }K }^{d(\alpha)} : i \in [0,K-k z-1]\}$.
We provide the proof of decodability for the delivery scheme presented in Algorithm \ref{algo2} depending upon whether the value of  $K-kz$ is even or odd separately in the next two subsections.
\subsection{Proof of Correctness of Algorithm \ref{algo2} when $K-kz$ is even} 
\label{proog algo2}

If $K-kz$ is even, out of the $K-kz$ sub-files which the user $U_{\alpha}$ needs to retrieve, we denote the first $\frac{K-kz}{2}$ consecutive sub-files, neighboring to the sub-files available with the user $U_{\alpha}$, as the set $\mathcal{P}_{\alpha}$, while the next $\frac{K-kz}{2}$ sub-files are denoted by the set $\mathcal{Q}_{\alpha}$, i.e.,
{\small
\begin{align}
\mathcal{P}_{\alpha}=& \left \{W_{(k \alpha+kz+i) \text{ mod } K}^{d(\alpha)} : i \in \left [0,\frac{K-kz}{2}-1 \right ]\right \} \label{eq even Palpha} \\ \mathcal{Q}_{\alpha} =& \left \{W_{(k\alpha+kz+i) \text{ mod } K}^{d(\alpha)} : i \in \left [\frac{K-kz}{2},K-k z-1\right ]\right \}.\label{eq even Qalpha}
\end{align}}
Hence, the user $U_{\alpha}$ needs to decode  all the sub-files in the set $\mathcal{P}_{\alpha} \cup \mathcal{Q}_{\alpha}.$ 
We rewrite the set $\mathcal{P}_{\alpha}$ as
in (\ref{eq; p alpha}), by changing the variable in the subscript of each sub-file from $i$ to $r=K-kz-i$. This is done to relate these sub-files to the ones present in the coded symbols in Algorithm \ref{algo2}.
{\small
\begin{align} \label{eq; p alpha}
\mathcal{P}_{\alpha} &= \left \{W_{(k \alpha+kz+K-kz-r) \text{ mod } K}^{d(\alpha)} : r \in \left [\frac{K-kz}{2}+1,K-kz \right ]\right \} \nonumber\\
&= \left \{W_{(k \alpha-r) \text{ mod } K}^{d(\alpha)} : r \in  \left [\frac{K-kz}{2}+1,K-kz \right ]\right \}.
\end{align}}
Similarly, we rewrite the set $\mathcal{Q}_{\alpha}$,
by changing the variable from $i$ to $r=i+1$, as
{\small
\begin{align} \label{ eq q alpha}
\mathcal{Q}_{\alpha} =  \left \{W_{(k \alpha+kz+r-1) \text{ mod } K}^{d(\alpha)}: r \in \left [\frac{K-kz}{2}+1,K-k z\right ]\right \}.
\end{align}}
We prove that  the user $U_{\alpha}$ can decode the sub-files in $\mathcal{P}_{\alpha}$ and $\mathcal{Q}_{\alpha}$ separately in Lemmas \ref{lem1} and \ref{lem2} respectively.

\begin{lem}
	\label{lem1}
	Each user $U_{\alpha}, \alpha \in [0,K-1]$, can decode all the sub-files in $\mathcal{P}_{\alpha}$ given by $(\ref{eq; p alpha})$, using the transmissions in Algorithm \ref{algo2}.
\end{lem}
\begin{proof}
	We need to prove that each user $U_{\alpha}$ can decode the sub-file $W_{(k \alpha-r) \text{ mod }K}^{d(\alpha)}$,  for each $r \in \left [\frac{K-kz}{2}+1 , K-kz \right ]$. We sub-divide this case further into two parts depending upon whether the value of $\lceil \frac{k z}{r} \rceil$ is odd or even. That is we divide the set $\mathcal{P}_{\alpha}$ into two disjoint subsets, $\mathcal{P}_{\alpha_1}$ and $\mathcal{P}_{\alpha_2}$, as in (\ref{eqn P1})  and (\ref{eqn P2}) respectively.
	\begin{strip}
	{\small
	\begin{align}
	\mathcal{P}_{\alpha_1} &= \left \{W_{(k \alpha-r) \text{ mod }K}^{d(\alpha)} : r \in  \left [\frac{K-kz}{2}+1, K-kz \right ],  \left \lceil \frac{k z}{r} \right \rceil  \text{ is odd } \right \} \label{eqn P1}\\ \mathcal{P}_{\alpha_2} &= \left \{W_{(k \alpha-r) \text{ mod }K}^{d(\alpha)} : r \in  \left [\frac{K-kz}{2}+1, K-kz \right ], \left  \lceil \frac{k z}{r} \right  \rceil  \text{ is even } \right \}.\label{eqn P2}
	\end{align}	 }
\end{strip}
	We prove that  the user $U_{\alpha}$ can decode all the sub-files in $\mathcal{P}_{\alpha_1}$ and $\mathcal{P}_{\alpha_2}$ separately in {\bf{Part 1}} and {\bf{Part 2}} respectively.
		
	\textit{{\bf{Part 1}}: Consider sub-files in $\mathcal{P}_{\alpha_1}$.}\\
	%
	We first fix some $r$, take the corresponding sub-file $W_{(k \alpha-r) \text{ mod }K}^{d(\alpha)}$ in $\mathcal{P}_{\alpha_1}$ and the corresponding value $p= 1+ \lceil \frac{k z}{r} \rceil $ in line 2 in Algorithm \ref{algo2}.
	Since $p$ is even, we need to prove that the user $U_{\alpha}$ can decode the sub-file $W_{(k \alpha-r) \text{ mod }K}^{d(\alpha)}$ by proving that the user can retrieve all the $\frac{p}{2}$ blocks, $\{W_{(k \alpha-r) \text{ mod }K,l}^{d(\alpha)}: l \in[0,\frac{p}{2}-1]\}$, of the sub-file $W_{(k \alpha-r) \text{ mod }K}^{d(\alpha)}$ (see line 11 in Algorithm \ref{algo2}).
	
	We prove that for each $l \in \left [0,\frac{p}{2}-1 \right ]$, the user can decode the block $W_{(k \alpha-r) \text{ mod }K,l}^{d(\alpha)}$ from the coded symbol $T_{(k \alpha-(l+1)r) \text{ mod }K}^{i}$ transmitted during the iteration $i=\left (r-\frac{K-kz}{2} \right )$. From line 13 in Algorithm \ref{algo2}, we have  
	{\small
	\begin{align} \label{eqn lemma 1 part 1}
	T_{(k \alpha-(l+1)r) \text{ mod }K}^{i}= \left (\bigoplus_{m \in \left [0,\frac{p}{2}-1 \right ]} W_{(k \alpha-(l+1-m)r) \text{ mod }K,m  }^{d(\pi^{-1}_{{i,1}}((k \alpha-(l+1-m)r) \text{ mod }K))} \right )
	\nonumber \\ \bigoplus
	\left (  \bigoplus_{m \in \left [\frac{p}{2},p-1\right ]} W_{(k \alpha-(l+1-m)r) \text{ mod }K,m-\frac{p}{2} }^{d(\pi^{-1}_{{i,2}}((k \alpha-(l+1-m)r) \text{ mod }K))} \right ).
	\end{align}	}
	The sub-file corresponding to $m=l$ in (\ref{eqn lemma 1 part 1}) is $W_{(k \alpha-r) \text{ mod }K}^{d(\pi_{i,1}^{-1}((k \alpha-r) \text{ mod }K))}$.
	Since $\pi_{i,1}^{-1}((k \alpha-r) \text{ mod }K) =\alpha$, the  sub-file corresponding to $m=l$ is $W_{(k \alpha-r) \text{ mod }K}^{d(\alpha )}$.
	Also, the $l^{th} $ block of the sub-file $W_{(k \alpha-r) \text{ mod }K}^{d(\alpha )}$ is present in (\ref{eqn lemma 1 part 1}). 
	
	In (\ref{eqn lemma 1 part 1}), note that the subscripts of the $p$ sub-files corresponding to $m \in \left [0, p-1 \right]$, belong to the set $\{(k\alpha -(l+1-m)r) \text{ mod } K :  m \in \left [0, p-1 \right] \} $. We need to show that all the sub-files corresponding to $m \in [0,p-1] \backslash l$, i.e., all the sub-files with the  subscripts in $\{(k\alpha -(l+1-m)r) \text{ mod } K :  m \in \left [0, p-1 \right] \backslash l\} $,  are available with the user $U_{\alpha}$. Equivalently, we need to prove that 
	
	{\small
	\begin{align} \label{eqn subscript}
	\{(k\alpha -(l+1-m)r ) \text{ mod } K:  m \in \left [0, p-1 \right] \backslash l\} \subseteq \nonumber \\ [k\alpha \text{ mod } K, (k\alpha +kz-1) \text{ mod } K].
	\end{align}}
	To prove this, we rewrite the LHS of (\ref{eqn subscript}) as
	{\small
	\begin{align}
	\underbrace{  \{(k\alpha -(l+1-m)r ) \text{ mod } K:  m \in \left [0, l-1 \right]\}}_{\mathcal{S}_1} \cup\nonumber \\\underbrace{ \{(k\alpha +(m-(l+1))r) \text{ mod } K :  m \in \left [ l+1,p-1 \right]\} }_{\mathcal{S}_2}.
	\end{align}}	
	We need to prove that both $\mathcal{S}_1$ and $\mathcal{S}_2$ belong to the RHS of (\ref{eqn subscript}).
	First, we consider the elements in the set $\mathcal{S}_1$.
	The minimum value of $m$ in the set $\mathcal{S}_1$ is $0$ and the maximum value is $l-1$. The element corresponding to $m=0$ in the set $\mathcal{S}_1$ is $(k \alpha -(l+1)r \text{ mod } K)$.
	The maximum value of $l$ is $\frac{p}{2}-1$ and the maximum value of $r$ is $K-kz$. Hence, 
	{\small
	\begin{align}
	(k \alpha -(l+1)r) \text{ mod } K =(k \alpha -lr-r) \text{ mod } K \nonumber\\\geq \left (k \alpha -\left (\frac{p}{2}-1 \right )r -K+kz \right ) \text{ mod } K \nonumber\\=\left (k \alpha +kz - \left (\frac{p}{2}-1 \right )r \right ) \text{ mod } K.\label{proof cache1}
	\end{align}}	
	Since $p = \lceil \frac{kz}{r} \rceil +1 $, we know that $(p-2)r <kz$. Therefore, $(\frac{p}{2}-1)r < kz$ and hence, $(k \alpha -(l+1)r) \text{ mod } K  \in  [k\alpha \text{ mod } K, (k\alpha +kz-1) \text{ mod } K]$. 
	
	The element corresponding to $m=l-1$ is $((k \alpha - 2r ) \text{ mod } K)$. 
	Since $r \geq \frac{K-kz}{2} +1 $, we know that $2r > K-kz$. Therefore, 
	{\small
	\begin{align} \label{proof cache2 }
	((k \alpha -2r ) \text{ mod } K)  &< ((k \alpha -K+kz ) \text{ mod } K)\\&= ((k \alpha +kz ) \text{ mod } K).     
	  \end{align}}
	Hence, $((k \alpha -2r ) \text{ mod } K) \in [k\alpha \text{ mod } K, (k\alpha +kz-1) \text{ mod } K]$. 
	
	Note that the elements in the set $\mathcal{S}_1$ are in the increasing order with $m$. Hence all the elements in $\mathcal{S}_1$ belong to the set $[ (k \alpha -(l+1)r) \text{ mod } K, (k \alpha -2r) \text{ mod } K]$. From (\ref{proof cache1}) and (\ref{proof cache2 }),  all the elements in $\mathcal{S}_1$ belong to the set $[\left (k \alpha +kz - \left (\frac{p}{2}-1 \right )r \right ) \text{ mod } K,(k \alpha +kz -1) \text{ mod } K]$.
	Therefore, $\mathcal{S}_1 \subseteq [k\alpha \text{ mod } K, (k\alpha +kz-1) \text{ mod } K]$.
		
	Now, consider the elements in the set $\mathcal{S}_2 $. The minimum value of $m$ in the set $\mathcal{S}_2$ is $l+1$ and the maximum value is $p-1$. The element corresponding to $m=l+1$ in the set $\mathcal{S}_2$ is $(k \alpha \text{ mod } K)$ while the element corresponding to $m=p-1$ is $((k \alpha + (p-2-l)r ) \text{ mod } K)$. Since $p = \lceil \frac{kz}{r} \rceil +1 $, we know that $(p-2)r <kz$. Therefore, $(p-2-l)r <kz$ and hence,  $((k \alpha + (p-2-l)r ) \text{ mod } K) \in [k\alpha \text{ mod } K, (k\alpha +kz-1) \text{ mod } K]$. The elements in the set $\mathcal{S}_2$ are in the increasing order with $m$ and the elements corresponding to the minimum as well as the maximum values of $m$ belong to the RHS of (\ref{eqn subscript}). Therefore, $\mathcal{S}_2 \subseteq [k\alpha \text{ mod } K, (k\alpha +kz-1) \text{ mod } K]$.
	
	So, the user can retrieve the block $W_{(k \alpha-r) \text{ mod }K,l}^{d(\alpha)}$ from the coded symbol $T_{(k \alpha-(l+1)r) \text{ mod }K}^{i}$ as shown in (\ref{eqn lemma 1 part 1 cache}) and (\ref{eqn lemma 1 part 1 cache1}).
	\begin{strip}
	{\small	
	\begin{align}
	T_{(k \alpha-(l+1)r) \text{ mod }K}^{i}=&  \left (W_{(k \alpha-r) \text{ mod }K,l}^{d(\pi_{i,1}^{-1}((k \alpha-r) \text{ mod }K) )}  \right ) \bigoplus \left (  \bigoplus_{m \in \left [0,\frac{p}{2}-1 \right ] \backslash l} W_{(k \alpha-(l+1-m)r) \text{ mod }K,m  }^{d(\pi^{-1}_{{i,1}}((k \alpha-(l+1-m)r) \text{ mod }K))} \right )\nonumber \\ &\hspace{4.6cm} \bigoplus \left ( \bigoplus_{m \in \left [\frac{p}{2},p-1\right ]} W_{(k \alpha-(l+1-m)r) \text{ mod }K,m-\frac{p}{2} }^{d(\pi^{-1}_{{i,2}}((k \alpha-(l+1-m)r) \text{ mod }K))} \right ) \label{eqn lemma 1 part 1 cache}
	\\
	=& \left (W_{(k \alpha-r) \text{ mod }K,l}^{d(\alpha)} \right ) \bigoplus  \underbrace{\left (   \bigoplus_{m \in \left [0,\frac{p}{2}-1 \right ] \backslash l} W_{(k \alpha-(l+1-m)r) \text{ mod }K,m  }^{d(\pi^{-1}_{{i,1}}((k \alpha-(l+1-m)r) \text{ mod }K))} \right )}_{\text{available at cache}} \bigoplus \nonumber\\& \hspace{3.8cm} \underbrace{\left (   \bigoplus_{m \in \left [\frac{p}{2},p-1\right ]} W_{(k \alpha-(l+1-m)r) \text{ mod }K,m-\frac{p}{2} }^{d(\pi^{-1}_{{i,2}}((k \alpha-(l+1-m)r) \text{ mod }K))}\right )}_{\text{available at cache}} 
	\label{eqn lemma 1 part 1 cache1}
	\end{align}}
\end{strip}
	\textit{{\bf{Part 2:}} Consider sub-files in $\mathcal{P}_{\alpha_2} $.}\\
	
	In this case also, we first fix some $r$, take the corresponding sub-file $W_{(k \alpha-r) \text{ mod }K}^{d(\alpha)}$ in $\mathcal{P}_{\alpha_2}$ and the corresponding $p= 1+ \lceil \frac{k z}{r} \rceil $. Since $p$ is odd, the user $U_{\alpha}$ needs to retrieve all the $p$ blocks, $\{W_{(k \alpha-r) \text{ mod }K,l}^{d(\alpha)}: l \in[0,p-1]\}$, of the sub-file $W_{(k \alpha-r) \text{ mod }K}^{d(\alpha)}$ in order to decode that sub-file.  
	
	We prove that  for each $l \in \left [0,\frac{p-3}{2} \right ]$, the user can decode $W_{(k \alpha-r) \text{ mod }K,l}^{d(\alpha)}$ from the coded symbol $T_{(k \alpha-(l+1)r) \text{ mod }K,1}^{i}$ transmitted during the iteration $i=r-\frac{K-kz}{2}$. Since, $K-kz$ is even and $p$ is odd, from line 18 in Algorithm \ref{algo2}, we have (\ref{eqn lemma 1 part 2 T1}).
	
	The sub-file corresponding to $m=l$ in (\ref{eqn lemma 1 part 2 T1}) is $W_{(k \alpha-r) \text{ mod }K}^{d(\pi_{i,1}^{-1}((k \alpha-r) \text{ mod }K))}$.
	Since $\pi_{i,1}^{-1}((k \alpha-r) \text{ mod }K) =\alpha$, we have $W_{(k \alpha-r) \text{ mod }K}^{d(\pi_{i,1}^{-1}((k \alpha-r) \text{ mod }K))} =W_{(k \alpha-r) \text{ mod }K}^{d(\alpha )}$.
	Also, the $l^{th} $ block of the sub-file $W_{(k \alpha-r) \text{ mod }K}^{d(\alpha )}$ is present in (\ref{eqn lemma 1 part 2 T1}).   Here also, all other sub-files in (\ref{eqn lemma 1 part 2 T1}) are available at the cache of User $U_{\alpha}$, the proof of which is similar to the one provided  in {\textbf{ Part 1}}. So, the user can retrieve the block  $W_{(k \alpha-r) \text{ mod }K,l}^{d(\alpha )}$ from the coded symbol $T_{(k \alpha-(l+1)r) \text{ mod }K,1}^{i}$ as shown in (\ref{eqn lemma 1 part 2 T1 cache}).
	\begin{strip}
	{\small
	\begin{align} 
	T_{(k \alpha-(l+1)r) \text{ mod }K,1}^{i}=&\left ( \bigoplus_{m \in \left [0,\frac{p-3}{2}\right  ]} W_{(k \alpha-(l+1-m)r) \text{ mod }K,m  }^{d(\pi^{-1}_{{i,1}}((k \alpha-(l+1-m)r) \text{ mod }K))} \right ) \bigoplus \nonumber\\& \hspace{1cm} \left ( \bigoplus_{m \in \left [\frac{p-1}{2},p-1\right ]} W_{(k \alpha-(l+1-m)r) \text{ mod }K,m-\frac{p-1}{2} }^{d(\pi^{-1}_{{i,2}}((k \alpha-(l+1-m)r) \text{ mod }K))} \right ) \label{eqn lemma 1 part 2 T1} \\
	=& \left (W_{(k \alpha-r) \text{ mod }K,l}^{d(\alpha)} \right ) \bigoplus \underbrace{ \left ( \bigoplus_{m \in \left [0,\frac{p-3}{2}\right ]\backslash l } W_{(k \alpha-(l+1-m)r) \text{ mod }K,m  }^{d(\pi^{-1}_{{i,1}}((k \alpha-(l+1-m)r) \text{ mod }K))} \right )}_{\text{available at cache}} \bigoplus \nonumber\\&
	\hspace{3.7cm}
	\underbrace{ \left (  \bigoplus_{m \in \left [\frac{p-1}{2},p-1\right ]}  W_{(k \alpha-(l+1-m)r) \text{ mod }K,m-\frac{p-1}{2} }^{d(\pi^{-1}_{{i,2}}((k \alpha-(l+1-m)r) \text{ mod }K))} \right )}_{\text{available at cache}} 
	\label{eqn lemma 1 part 2 T1 cache}
	\end{align}}
\end{strip}
	Now, we prove that for each $l \in \left [\frac{p-1}{2},p-1\right ]$, the user can decode $W_{(k \alpha-r) \text{ mod }K,l}^{d(\alpha)}$ from the coded symbol  $T_{(k \alpha-\left ( l-\frac{p-1}{2} +1\right )r) \text{ mod }K,2}^{i}$ transmitted during the iteration $i=r-\frac{K-kz}{2}$, from line 18 in Algorithm \ref{algo2}.
	Here, the sub-file corresponding to $m= (l-\frac{p}{2})$ in (\ref{eqn lemma 1 part 2 T2}) is $W_{(k \alpha-r) \text{ mod }K}^{d(\pi_{i,1}^{-1}((k \alpha-r) \text{ mod }K))}$.
	Since $\pi_{i,1}^{-1}((k \alpha-r) \text{ mod }K) =\alpha$, we have $W_{(k \alpha-r) \text{ mod }K}^{d(\pi_{i,1}^{-1}((k \alpha-r) \text{ mod }K))} =W_{(k \alpha-r) \text{ mod }K}^{d(\alpha )}$.
	Also, the $l^{th} $ block of the sub-file $W_{(k \alpha-r) \text{ mod }K}^{d(\alpha )}$ is present in (\ref{eqn lemma 1 part 2 T2}).   
	All other sub-files in (\ref{eqn lemma 1 part 2 T2}) are available at the cache of User $U_{\alpha}$, the proof of which is similar to that provided in {\textbf{ Part 1}}. 
	So, the user can retrieve the block  $W_{(k \alpha-r) \text{ mod }K,l}^{d(\alpha )}$ from the coded symbol $T_{(k \alpha-\left ( l-\frac{p-1}{2} +1\right )r) \text{ mod }K,2}^{i}$ as shown in (\ref{eqn lemma 1 part 2 T2 cache}).
	\begin{strip}
	{\small
	\begin{align}
	T_{(k \alpha-\left ( l-\frac{p-1}{2} +1\right )r) \text{ mod }K,2}^{i}=&\left ( \bigoplus_{m \in \left [0,\frac{p-1}{2} \right ]} W_{(k \alpha-\left ( l-\frac{p-1}{2} +1-m\right )r) \text{ mod }K,m +\frac{p-1}{2} }^{d(\pi^{-1}_{i,1} ( (k \alpha-\left ( l-\frac{p-1}{2} +1-m\right )r) \text{ mod }K))} \right ) \bigoplus \nonumber \\
	&\hspace{1cm} \left ( \bigoplus_{m \in \left [\frac{p+1}{2},p-1\right ]} W_{(k \alpha-\left ( l-\frac{p-1}{2} + 1-m\right )r) \text{ mod }K,m }^{d(\pi^{-1}_{i,2} ( (k \alpha-\left ( l-\frac{p-1}{2} +1-m\right )r) \text{ mod }K))} \right ) \label{eqn lemma 1 part 2 T2}\\
	=&W_{(k \alpha-r) \text{ mod }K,l}^{d(\alpha)}  \bigoplus \underbrace{ \left (  \bigoplus_{m  \in \left [0,\frac{p-1}{2} \right ] \backslash (l-\frac{p-1}{2})} W_{(k \alpha-\left ( l-\frac{p-1}{2} +1-m\right )r) \text{ mod }K,m }^{d(\pi^{-1}_{i,1} ((k \alpha-\left ( l-\frac{p-1}{2} +1-m\right )r) \text{ mod }K ))} \right
		)}_{\text{available at cache}} \nonumber\\
	&\hspace{2.3cm} \bigoplus \underbrace{ \left (   \bigoplus_{m \in \left [\frac{p+1}{2},p-1 \right ]} W_{(k \alpha-\left ( l-\frac{p-1}{2} +1-m\right )r) \text{ mod }K,m }^{d(\pi^{-1}_{i,2} ((k \alpha-\left ( l-\frac{p-1}{2} +1-m\right )r) \text{ mod }K ))}\right )}_{\text{available at cache}} 
	\label{eqn lemma 1 part 2 T2 cache}
	\end{align}}	
\end{strip}
	In short, the user $U_{\alpha} $ can decode the sub-file $W_{(k \alpha-r) \text{ mod }K}^{d(\alpha)} $, since it has retrieved all the blocks of sub-files corresponding to $W_{(k \alpha-r) \text{ mod }K}^{d(\alpha)} $. This completes the proof for the set $\mathcal{P}_{\alpha}$.
\end{proof}
\begin{lem}
	\label{lem2}
	Each user $U_{\alpha}, \alpha \in \{0,1,\ldots ,\alpha-1\}$, can decode all the sub-files  in $\mathcal{Q}_{\alpha} $ as given in $(\ref{ eq q alpha})$ using the transmissions in Algorithm \ref{algo2}.
\end{lem}
\begin{proof}
	We need to prove that each user $U_{\alpha}$ can decode  the sub-file $W_{ (k \alpha+kz+r-1) \text{ mod }K}^{d(\alpha)}$, for each
	$ r \in\left [\frac{K-k z}{2}+1,K-kz \right ] $.
	In this case also, we sub-divide this case further into two parts depending upon whether the value of $\lceil \frac{k z}{r} \rceil$ is odd or even. We divide the set $\mathcal{Q}_{\alpha}$ into two disjoint subsets, $\mathcal{Q}_{\alpha_1}$ and $\mathcal{Q}_{\alpha_2}$, as in (\ref{eqn Q1}) and (\ref{eqn Q2}) respectively.
	\begin{strip}
		\begin{equation}\label{eqn Q1}
	\mathcal{Q}_{\alpha_1} = \left \{W_{ (k\alpha+kz+r-1) \text{ mod }K}^{d(\alpha)} : r \in\left [\frac{K-k z}{2}+1,K-kz \right ] , \lceil \frac{k z}{r} \rceil  \text{ is odd } \right \}
		\end{equation}
		\begin{equation}\label{eqn Q2}
	\mathcal{Q}_{\alpha_2} = \left \{W_{ (k\alpha+kz+r-1) \text{ mod }K}^{d(\alpha)} : r \in\left [\frac{K-k z}{2}+1,K-kz \right ] , \lceil \frac{k z}{r} \rceil  \text{ is even } \right \}
		\end{equation}
		\end{strip}
	We prove that  the $U_{\alpha}$ can decode all the sub-files in $\mathcal{Q}_{\alpha_1}$ and $\mathcal{Q}_{\alpha_2}$ separately in {\bf{Part 3}} and {\bf{Part 4}} respectively.
	
	\textit{{\bf{Part 3}}: Consider sub-files in $\mathcal{Q}_{\alpha_1} $.}\\
	As in {\bf{Part 1}}, in this case also we first fix some $r$, take the corresponding sub-file $W_{ (k \alpha+kz+r-1) \text{ mod }K}^{d(\alpha)}$ in $\mathcal{Q}_{\alpha_1}$ and the corresponding $p= 1+ \lceil \frac{k z}{r} \rceil $. Since $p$ is even, the user $U_{\alpha}$ needs to retrieve all the $\frac{p}{2}$ blocks, $\{W_{ (k \alpha+kz+r-1) \text{ mod }K,l}^{d(\alpha)}: l \in[0,\frac{p}{2}-1]\}$, of the sub-file $W_{ (k \alpha+kz+r-1) \text{ mod }K}^{d(\alpha)}$ in order to decode that sub-file.  
	
	We prove that for each $l \in \left [0,\frac{p}{2}-1 \right ]$, the user can decode $W_{ (k \alpha+kz+r-1) \text{ mod }K}^{d(\alpha)}$ from the coded symbol $T_{j  }^{i}$, where $j=(k\alpha+ kz+r-1-\left (l+\frac{p}{2}\right )r) \text{ mod }K$, transmitted during the iteration $i=r-\frac{K-kz}{2}$. From line 13 in Algorithm \ref{algo2}, we have,
	\begin{align}
	T_j^{i} = \left (\bigoplus_{m =0}^{\frac{p}{2}-1} W_{(mr+j) \text{ mod } K,m}^{d\left(\pi^{-1}_{{i,1}}((mr+j) \text{ mod } K)\right) }  \right ) 
	\bigoplus \\\hspace{1cm} \left (\bigoplus_{m =\frac{p}{2}}^{p-1} W_{(mr+j) \text{ mod } K,m-\frac{p}{2}}^{d\left(\pi^{-1}_{{i,2}}((mr+j) \text{ mod } K)\right) }  \right ).
	\label{eqn lemma 2 part 3}
	\end{align}
	The sub-file corresponding to $m=l+\frac{p}{2}$ in (\ref{eqn lemma 2 part 3}) is $W_{(k \alpha+kz+r-1) \text{ mod }K}^{d(\pi_{i,2}^{-1}((k \alpha+kz+r-1) \text{ mod }K))}$.
	Since $\pi_{i,2}^{-1}((k \alpha+kz+r-1) \text{ mod }K) =\alpha$, we have $W_{(k \alpha+kz+r-1) \text{ mod }K}^{d(\pi_{i,2}^{-1}((k \alpha+kz+r-1) \text{ mod }K))} =W_{(k \alpha+kz+r-1) \text{ mod }K}^{d(\alpha )}$.
	Also, the $l^{th} $ block of the sub-file $W_{(k \alpha+kz+r-1) \text{ mod }K}^{d(\alpha )}$ is present in (\ref{eqn lemma 2 part 3}).   Here also, the user can retrieve the block  $W_{(k \alpha+kz+r-1) \text{ mod }K,l}^{d(\alpha )}$ from the coded symbol $T_{j}^{i}$ since, all other sub-files in (\ref{eqn lemma 2 part 3}) are available at its cache (for the same reason stated for the case discussed in {\textbf{ Part 1}}).	
	
	\textit{{\bf{Part 4:}} Consider sub-files in $\mathcal{Q}_{\alpha_2}$.}\\
	
	We first fix some $r$, take the corresponding sub-file $W_{ (k \alpha+kz+r-1) \text{ mod }K}^{d(\alpha)}$ in $\mathcal{Q}_{\alpha_2}$ and the corresponding $p= 1+ \lceil \frac{k z}{r} \rceil $. So the user $U_{\alpha}$ needs to retrieve all the $p$ blocks, $\{W_{ (k \alpha+kz+r-1) \text{ mod }K,l}^{d(\alpha)}: l \in[0,p-1]\}$, of the sub-file $W_{ (k \alpha+kz+r-1) \text{ mod }K}^{d(\alpha)}$ in order to decode that sub-file.  
	
	We prove that  for each $l \in \left [0,\frac{p-1}{2} \right ]$, the user can decode $W_{ (k \alpha+kz+r-1) \text{ mod }K}^{d(\alpha)}$ from the coded symbol  $T_{j  ,1}^{i}$, where $j=((k \alpha+ kz+r-1)-\left (l+\frac{p-1}{2}\right )r) \text{ mod }K$, transmitted during the iteration $i=r-\frac{K-kz}{2}$. From  line 18 in Algorithm \ref{algo2}, we have
	{\small	
	\begin{align}
	T_{j,1}^{i} = \left (\bigoplus_{m=0}^{\frac{p-3}{2}} W_{(mr+j) \text{ mod } K,m}^{d\left (\pi^{-1}_{{i,1}}((mr+j) \text{ mod } K)\right) } \right ) \bigoplus\\\hspace{1cm} \left (
	\bigoplus_{m= \frac{p-1}{2}}^{p-1} W_{(mr+j) \text{ mod } K,m-\frac{p-1}{2}}^{d \left (\pi^{-1}_{{i,2}}((mr+j) \text{ mod } K)\right)  } \right )
	\label{eqn lemma 2 part 4 T1}
	\end{align}}	
	The sub-file corresponding to $m=l+\frac{p-1}{2}$ in (\ref{eqn lemma 2 part 4 T1}) is $W_{(k \alpha+kz+r-1) \text{ mod }K}^{d(\pi_{i,2}^{-1}((k \alpha+kz+r-1) \text{ mod }K))}$.
	Since $\pi_{i,2}^{-1}((k \alpha+kz+r-1) \text{ mod }K) =\alpha$, we have $W_{(k \alpha+kz+r-1) \text{ mod }K}^{d(\pi_{i,2}^{-1}((k \alpha+kz+r-1) \text{ mod }K))} =W_{(k \alpha+kz+r-1) \text{ mod }K}^{d(\alpha )}$.
	Also, the $l^{th} $ block of the sub-file $W_{(k \alpha+kz+r-1) \text{ mod }K}^{d(\alpha )}$ is present in (\ref{eqn lemma 2 part 4 T1}).  For the same reason we had stated for the case discussed in {\textbf{ Part 1}}, the user $U_{\alpha}$ can retrieve the block  $W_{(k \alpha+kz+r-1) \text{ mod }K,l}^{d(\alpha )}$ from the coded symbol $T_{j,1}^{i}$ since, all other sub-files in (\ref{eqn lemma 2 part 4 T1}) are available at its cache.
	
	Now, we prove that for each $l \in \left [\frac{p+1}{2},p-1\right ]$, the user can decode $W_{ (k \alpha+kz+r-1 ) \text{ mod }K,l}^{d(\alpha)}$ from the coded symbol  $T_{j ,2}^{i}$, where $j=(k\alpha+ kz-1-(l-1)r) \text{ mod } K$, transmitted during the iteration $i=r-\frac{K-kz}{2}$. From  line 18 in Algorithm \ref{algo2}, we have
	{\small
	\begin{align}
	T_{j,2}^{i} = \left (\bigoplus_{m =0}^{\frac{p-1}{2}} W_{(mr+j) \text{ mod } K,\frac{p-1}{2}+m}^{d\left(\pi^{-1}_{{i,1}}((mr+j) \text{ mod } K)\right) } \right ) \bigoplus \\\hspace{1cm} \left (
	\bigoplus_{m= \frac{p+1}{2}}^{p-1} W_{(mr+j) \text{ mod } K,m}^{d\left(\pi^{-1}_{{i,2}}((mr+j) \text{ mod } K)\right)  } \right ).
	\label{eqn lemma 2 part 4 T2}
	\end{align}}
	The sub-file corresponding to $m=l$ in (\ref{eqn lemma 2 part 4 T2}) is $W_{(k \alpha+kz+r-1) \text{ mod }K}^{d(\pi_{i,2}^{-1}((k \alpha+kz+r-1) \text{ mod }K))}$.
	Since $\pi_{i,2}^{-1}((k \alpha+kz+r-1) \text{ mod }K) =\alpha$, we have $W_{(k \alpha+kz+r-1) \text{ mod }K}^{d(\pi_{i,2}^{-1}((k \alpha+kz+r-1) \text{ mod }K))} =W_{(k \alpha+kz+r-1) \text{ mod }K}^{d(\alpha )}$.
	Also, the $l^{th} $ block of the sub-file $W_{(k \alpha+kz+r-1) \text{ mod }K}^{d(\alpha )}$ is present in (\ref{eqn lemma 2 part 4 T2}).  For the same reason we had stated for the case discussed in {\textbf{ Part 1}}, the user $U_{\alpha}$ can retrieve the block  $W_{(k \alpha+kz+r-1) \text{ mod }K,l}^{d(\alpha )}$ from the coded symbol $T_{j,2}^{i}$ since, all other sub-files in (\ref{eqn lemma 2 part 4 T2}) are available at its cache.
	
	In short, the user $U_{\alpha} $ can decode the sub-file $W_{ (k \alpha+kz+r-1) \text{ mod }K}^{d(\alpha)}$, since it has retrieved all the blocks of the sub-file corresponding to $W_{ (k \alpha+kz+r-1) \text{ mod }K}^{d(\alpha)}$ This completes the proof for the set $\mathcal{Q}_{\alpha}$.
\end{proof}
From Lemmas \ref{lem1} and \ref{lem2}, when $K-k z$ is even, each user $U_{\alpha}$ can decode all the sub-files, $\{W_{(k \alpha+k z+i) \text{ mod }K}^{d(\alpha)}, \forall i \in [0,K-k z-1] \}$, corresponding to its demanded file $W^{d(\alpha)}$, which are not available in its cache.
\subsection{Proof of Correctness of Algorithm \ref{algo2} when $K-kz$ is odd.}
\label{proof algo2}
If $K-kz$ is odd, out of the $K-kz$ sub-files which the user $U_{\alpha}$ needs to retrieve, we denote the first $\frac{K-kz-1}{2}$ consecutive sub-files, neighboring to the sub-files available with the user $U_{\alpha}$, as the set $\mathcal{P}'_{\alpha}$, the next sub-file as the set $\mathcal{O}_{\alpha}$ and the remaining $\frac{K-kz-1}{2}$ sub-files as the set $\mathcal{Q}'_{\alpha}$, i.e.,
{\small
\begin{align}
\mathcal{P}'_{\alpha}=& \left \{W_{(k \alpha+kz+i) \text{ mod } K}^{d(\alpha)} : i \in \left [0,\frac{K-kz-1}{2}-1 \right ]\right \} \label{eq odd Palpha} \\
\mathcal{O}_{\alpha} =& \left \{W_{(k \alpha+kz+i) \text{ mod } K}^{d(\alpha)} : i = \frac{K-kz-1}{2} \right \} \label{eq odd Oalpha} \\
\mathcal{Q}'_{\alpha} =& \left \{W_{(k\alpha+kz+i) \text{ mod } K}^{d(\alpha)} : i \in \left [\frac{K-kz-1}{2}+1,K-k z-1\right ]\right \}.\label{eq odd Qalpha}
\end{align}}
Hence, the user $U_{\alpha}$ needs to decode  all the sub-files in the set $\mathcal{P}'_{\alpha} \cup \mathcal{O}_{\alpha} \cup \mathcal{Q}'_{\alpha}.$ 
We rewrite the set $\mathcal{P}'_{\alpha}$ as in (\ref{eq; p alpha odd}), by changing the variable in the subscript of each sub-file from $i$ to $r=K-kz-i$. This is done to relate these sub-files to the ones present in the coded symbols in Algorithm \ref{algo2}.
{\small
\begin{align} \label{eq; p alpha odd}
\mathcal{P}'_{\alpha} &= \left \{W_{(k \alpha+kz+K-kz-r) \text{ mod } K}^{d(\alpha)} : r \in \left [\frac{K-kz+1}{2}+1,K-kz \right ]\right \} \nonumber\\
&= \left \{W_{(k \alpha-r) \text{ mod } K}^{d(\alpha)} : r \in  \left [\frac{K-kz+1}{2}+1,K-kz \right ]\right \}.
\end{align}}
We rewrite the set $\mathcal{O}_{\alpha}$,
by changing the variable from $i$ to $r=K-kz-i$, as
{\small
\begin{align} \label{eq O alpha odd}
\mathcal{O}_{\alpha} &= \left \{W_{(k \alpha+kz+K-kz-r) \text{ mod } K}^{d(\alpha)} : r =\frac{K-kz+1}{2}\right \} \nonumber\\
&= \left \{W_{(k \alpha-r) \text{ mod } K}^{d(\alpha)} :r =\frac{K-kz+1}{2}\right \}.
\end{align}}
Similarly, we rewrite the set $\mathcal{Q}'_{\alpha}$,
by changing the variable from $i$ to $r=i+1$, as
{\small
\begin{align} \label{eq q alpha odd}
\mathcal{Q}'_{\alpha} =  \left \{W_{(k \alpha+kz+r-1) \text{ mod } K}^{d(\alpha)}: r \in \left [\frac{K-kz+1}{2}+1,K-k z\right ]\right \}.
\end{align}}

We prove that  the user $U_{\alpha}$ can decode the sub-files in $\mathcal{P}'_{\alpha}$, $\mathcal{Q}'_{\alpha}$  and $\mathcal{O}_{\alpha}$ separately in Lemmas \ref{lem3}, \ref{lem4} and \ref{lem5} respectively.

\begin{lem}
	\label{lem3}
	Each user $U_{\alpha}, \alpha \in [0 ,K-1]$, can decode all the sub-files in $\mathcal{P}'_{\alpha}$ given by (\ref{eq; p alpha odd}), using the transmissions in Algorithm \ref{algo2}.
\end{lem}
\begin{lem}
	\label{lem4}
	Each user $U_{\alpha}, \alpha \in [0,K-1]$, can decode all the sub-files in $\mathcal{Q}'_{\alpha}$ given by (\ref{eq q alpha odd}), using the transmissions in Algorithm \ref{algo2}.
\end{lem}
The proofs of  Lemmas \ref{lem3} and \ref{lem4} are similar to the proofs of Lemmas \ref{lem1} and \ref{lem2} respectively. 
\begin{lem}
	\label{lem5}
	Each user $U_{\alpha}, \alpha \in [0 ,K-1]$, can decode the sub-file in  $\mathcal{O}_{\alpha}$ given by (\ref{eq O alpha odd}), using the transmissions in Algorithm \ref{algo2}.
\end{lem}
\begin{proof}
	Each user $U_{\alpha}$ uses the coded symbols obtained during iteration $1$ to decode the sub-file $W_{(k \alpha-r)\text{ mod } K}^{d(\alpha)}$. 
	The user $U_{\alpha}$ needs to retrieve all the $p$ blocks, $\{W_{(k \alpha-r)\text{ mod } K,l}^{d(\alpha)}: l \in [0,p-1]\}$, of the sub-file $W_{(k \alpha-r)\text{ mod } K}^{d(\alpha)}$ in order to decode that sub-file (see line 5 in Algorithm \ref{algo2}). For each $l \in[0,p-1]$, the user $U_{\alpha}$ retrieves the block $W_{(k \alpha-r)\text{ mod } K,l}^{d(\alpha)}$, from the coded symbol $T_{j \text{ mod } K}^{1}$, where $j=k \alpha-(l+1)r$.             
	From line 7 in Algorithm \ref{algo2}, we have
	{\small
	\begin{align}
	T_{j \text{ mod } K}^{1} &= \bigoplus_{m \in [0,p-1] } W_{(k \alpha+(m-l-1)r)\text{ mod } K,m}^{d{(\pi_{1,1}^{-1}((k \alpha+(m-l-1)r)\text{ mod } K))}} \\
	&= W_{(k \alpha-r)\text{ mod } K,l}^{d{(\pi_{1,1}^{-1}((k \alpha-r)\text{ mod } K))}} \nonumber\\ &\hspace{1cm}\bigoplus_{m \in [0,p-1] \backslash l} W_{(k \alpha+(m-l-1)r)\text{ mod } K,m}^{d{(\pi_{1,1}^{-1}((k \alpha+(m-l-1)r)\text{ mod } K))}} \\
	&= W_{(k \alpha-r)\text{ mod } K,l}^{d(\alpha)} \bigoplus \nonumber\\&\hspace{1cm} \underbrace{  \left ( \bigoplus_{m \in [0,p-1] \backslash l} W_{(k \alpha+(m-l-1)r)\text{ mod } K,m}^{d{(\pi_{1,1}^{-1}((k \alpha+(m-l-1)r)\text{ mod } K))}} \right )}_{\text{available at cache}} . \label{eqn K odd}
	\end{align}}
	So, the user $U_{\alpha} $ can decode $W_{(k \alpha-r)\text{ mod } K,l}^{d(\alpha)},  $ for each $l \in[0,p-1]$, from $T_{(k \alpha-(l+1)r)\text{ mod } K}^{1}$, since all other blocks of sub-files are available at its cache (as shown in (\ref{eqn K odd})), the proof of which is similar to the one provided in {\textbf{Part 1} } in Lemma \ref{lem1}. Hence it can retrieve the sub-file $W_{(k \alpha-r)\text{ mod } K}^{d(\alpha)}$. This completes the proof for the set $\mathcal{O}_{\alpha}$.
\end{proof}
From Lemmas \ref{lem3}, \ref{lem4} and \ref{lem5}, when $K -k z$ is odd, each user $U_{\alpha}$ can decode all the sub-files, $\{W_{(k \alpha+k z+i)\text{ mod } K}^{d(\alpha)}, \forall i \in [0,K-k z-1] \}$, corresponding to its demanded file $W^{d(\alpha)}$, which are not available in its cache.
\end{document}